\documentclass[lettersize,journal]{IEEEtran}
\usepackage{amsmath,amsthm}

\usepackage[linesnumbered,ruled,vlined]{algorithm2e}

\newtheorem{theorem}{Theorem}
\newtheorem{corollary}{Corollary}

\usepackage[linesnumbered,ruled,vlined]{algorithm2e}
\usepackage{graphicx}
\usepackage{makecell}
\usepackage{multirow} 
\usepackage{tikz}
\usepackage{amsmath,amsthm}
\usepackage{subfigure}
\usepackage{caption}
\usepackage{diagbox}
\usepackage{bm}
\usepackage{amsfonts}
\usepackage{amssymb}
\usepackage{verbatim}
\usepackage{url}
\usepackage{booktabs}
\usepackage{hyperref}
\usepackage{cite}
\usepackage{color}
\usepackage{xcolor}
\usepackage{colortbl} 

\usepackage{filecontents}
\usepackage{mathrsfs}
\usepackage{mdframed}

\usepackage[most]{tcolorbox}
\usepackage{enumitem} 
\usepackage[normalem]{ulem} 

 \newenvironment{packeditemize}{
	\begin{list}{$\bullet$}{
			\setlength{\labelwidth}{4pt}
			\setlength{\itemsep}{0pt}
			\setlength{\leftmargin}{\labelwidth}
			\addtolength{\leftmargin}{\labelsep}
			\setlength{\parindent}{0pt}
			\setlength{\listparindent}{\parindent}
			\setlength{\parsep}{0pt}
			\setlength{\topsep}{1pt}}}{\end{list}}

\AtBeginDocument{%
  }

\newcommand{\change}[1]{\textcolor{black}{#1}}

\begin{document}

\title{BlockFUL: Enabling Unlearning in Blockchained Federated Learning}



\author{Xiao~Liu,
        Mingyuan~Li,
        Guangsheng~Yu,
        Xu~Wang,
        Wei~Ni,
        Lixiang~Li,
        Haipeng~Peng,
        and~Ren~Ping~Liu
        
\IEEEcompsocitemizethanks{

\IEEEcompsocthanksitem This work was supported in part by the National Key R\&D Program of China (Grant Nos. 2024YFB2906503 and 2024YFB2906500), in part by the National Natural Science Foundation of China (Grant No. 62032002), and in part by the Doctoral Student Innovation Fund (Grant No. CX2023217). Xiao Liu and Mingyuan Li contributed equally to this work and are co-first authors. 

\IEEEcompsocthanksitem Corresponding authors: Guangsheng Yu; Lixiang Li

\IEEEcompsocthanksitem X. Liu, M. Li, L. Li, and H. Peng are with the Information Security Center, State Key Laboratory of Networking and Switching Technology, Beijing University of Posts and Telecommunications, Beijing 100876, China. \protect
E-mail: \{liuxiao68, henryli\_i, lixiang, penghaipeng\}@bupt.edu.cn
\IEEEcompsocthanksitem G. Yu, X. Wang, W. Ni, and R. P. Liu are with the Global Big Data Technologies Centre, University of Technology Sydney, Australia. \protect
E-mail: \{guangsheng.yu, xu.wang, wei.ni,  renping.liu\}@uts.edu.au

}
}


\maketitle

\begin{abstract}

Unlearning in Federated Learning (FL) presents significant challenges, as models grow and evolve with complex inheritance relationships. This complexity is amplified when blockchain is employed to ensure the integrity and traceability of FL, where the need to edit multiple interlinked blockchain records and update all inherited models complicates the process.
In this paper, we introduce Blockchained Federated Unlearning (BlockFUL), a novel framework with a dual-chain structure—comprising a live chain and an archive chain—for enabling unlearning capabilities within Blockchained FL. BlockFUL introduces two new unlearning paradigms, i.e., parallel and sequential paradigms, which can be effectively implemented through gradient-ascent-based and re-training-based unlearning methods. These methods enhance the unlearning process across multiple inherited models by enabling efficient consensus operations and reducing computational costs. 
Our extensive experiments validate that these methods effectively reduce data dependency and operational overhead, thereby boosting the overall performance of unlearning inherited models within BlockFUL on CIFAR-10 and Fashion-MNIST datasets using AlexNet, ResNet18, and MobileNetV2 models.
\end{abstract}


\section{Introduction}\label{section:1}

Blockchained Federated Learning (FL) integrates blockchain with FL to enhance model trustworthiness and regulatory compliance, and support tokenized model ownerships with decentralized model governance. 
Its certified-by-blockchain mechanism ensures that model lineage and updates are verifiable and tamper-proof, enabling secure model sharing and facilitating commercialization. In healthcare, Blockchained FL facilitates disease diagnosis models, extending to applications like personalized treatments and drug development while ensuring traceability and auditability~\cite{uddin2021blockchain}. 
In finance, it supports financial prediction models and derivative products, maintaining data privacy, transparency, and regulatory adherence~\cite{li2020blockchain}. Additionally, Blockchained FL aligns with various FL structures—including traditional~\cite{krauss2024automatic}, multi-layer~\cite{chu2024multi}, multi-server~\cite{nguyen2022latency}, and decentralized FL~\cite{yu2023ironforge}—where model inheritance relationships naturally form a directed acyclic graph (DAG), reinforcing its relevance in large-scale collaborative learning.

While helping prevent FL models from being tampered with~\cite{li2020blockchain} and provides traceability~\cite{uddin2021blockchain},
using blockchains stops the models from being rectified (when needed, e.g., some training data contains sensitive information or is later identified as questionable or contaminated). 
In addition, any user is entitled to request to eliminate the impact of its personal training data on machine learning models~\cite{halimi2022federated,yuan2023federated}. In this sense, machine unlearning~\cite{chundawat2023zero} and the capability of editing and updating blocks in Blockchained FL~\cite{yu2020enabling} become relevant.

An advanced one among existing machine unlearning methods is Sharded, Isolated, Sliced, and Aggregated (SISA)~\cite{bourtoule2021machine}, which divides data into independent slices and trains them separately. Only the affected model slice needs to be re-trained to improve unlearning efficiency. 
Unfortunately, the inheritance relationships among models prevent SISA unlearning from being executed in parallel. This is due to the nature of incremental learning, where knowledge is inherited along with the reference between models. This nature renders SISA incapable of isolating unlearning processes since models need to be broadcast to all clients in any round~\cite{koch2023no}. With SISA ruled out, alternatives such as  re-training~\cite{liu2020federated,halimi2022federated,liu2022right} and gradient ascent~\cite{liu2022backdoor} come to the fore.
Re-training uses an updated dataset to re-train the entire model. Gradient ascent adjusts the model weights to generate larger errors on known data, thereby reducing its reliance on previously learned content and achieving ``forgetfulness''.

Integrating blockchains with FL under model inheritance relations can incur significant overhead in the unlearning processes. One reason is that although primitives, such as the Chameleon Hash (CH), have been considered for conditionally editing blockchains~\cite{jia2022redactable} at non-negligible costs with complex operations, historical models are recorded on-chain for certification purposes throughout iterations. Unlearning may require editing multiple related blockchain records to remove just one class of knowledge, which significantly increases the complexity of the unlearning operations. Another reason arises from the model inheritance, where a child model inherits and extends its parent model's characteristics, structure, and parameters. When performing an unlearning operation, the requested model and all its child models need to be updated.

To this end, two critical Research Questions (RQs) need to be addressed. 

\begin{packeditemize}
    \item 
\textbf{RQ1:} \textit{\change{How can a Blockchained FL framework unlearn historical models along with all inherited models while maintaining the traceability and integrity of models?}}

    \item 
\textbf{RQ2:} \textit{What unlearning methods can be adapted in a Blockchained FL system, and what are their associated performance and costs?}
   
\end{packeditemize}

In response to these two RQs, we propose a new framework, named BlockFUL, to empower Blockchained FL with unlearning capability. In this framework, users can delete data that they wish to remove or deem questionable, and update their models at an acceptable cost, without compromising the immutability of the blockchain, affecting the utility of inherited models, or altering their network structure.

\smallskip
\noindent\textbf{Contributions.}
This paper presents a novel Blockchained Federated Unlearning (BlockFUL) framework, which ensures model inheritance within generic FL systems and supports effective unlearning without compromising model interrelationships. BlockFUL features a dual-chain architecture—comprising a live chain and an archive chain—to provide users with appropriate access for unlearning. Designed for flexibility, BlockFUL supports a plug-and-play approach.

The key contributions of this paper are as follows: 
\begin{packeditemize}
  \item[$\bullet$]
  We propose the BlockFUL framework, the first comprehensive framework enabling unlearning within Blockchained FL. BlockFUL adopts a dual-chain architecture: a live chain and an archive chain. This architecture ensures the traceability, compliance, and integrity of model updates.

  \item[$\bullet$]
   \change{
  To embody and operationalize BlockFUL, we design two distinct unlearning paradigms:
  Parallel and sequential unlearning. Parallel unlearning enables concurrent updating of multiple models using a single consensus, significantly reducing consensus frequency and computational overhead. Sequential unlearning 
  serializes unlearning operations to prioritize the precision and order of model revisions.}

\item[$\bullet$]
\change{
We validate the feasibility and practicality of BlockFUL for both paradigms by demonstrating its compatibility with existing primitives, including CH-based redactable blockchain~\cite{jia2022redactable}, gradient ascent unlearning~\cite{liu2022backdoor} and re-training unlearning~\cite{liu2020federated}. We also analytically prove the convergence of the unlearning processes and quantify the computation, communication, and consensus costs for both parallel and sequential paradigms.}

 \end{packeditemize}

Experimental results confirm that sequential unlearning with re-training is highly effective, achieving nearly zero accuracy for unlearned data and \textbf{94.71\%} accuracy for retained data across various models and datasets while preserving model inheritance. However, it can be time-intensive for deep inheritance chains and requires full client participation. In contrast, parallel unlearning with gradient ascent shows varying effectiveness depending on model type and inheritance depth, with unlearning effectiveness ranging from \textbf{1.67\%} in AlexNet to \textbf{77.67\%} for retained data.

The rest of this paper is organized as follows. The related works are reviewed in Section~\ref{section:2}. The preliminary in Section~\ref{Preliminary}. The proposed BlockFUL framework is provided in Section~\ref{sec:system model}. The designed dual-chain architecture is presented in Section~\ref{sec:3}, followed by two unlearning paradigms in BlockFUL introduced in Section~\ref{section:4}. Section~\ref{section:5} presents the comparative experiments between the gradient ascent and re-training methods implemented in the new BlockFUL framework. Section~\ref{section:6}  concludes this work.

\section{Related Work}\label{section:2}

\subsection{Federated Unlearning}
Existing federated unlearning (FUL) studies focus primarily on parameter adjustment and re-training processes. 

\subsubsection{Parameter adjustment}
Halimi et al.~\cite{halimi2022federated} reversed a learning process by training the model to maximize local empirical losses and executed deletion of client contributions using Projected Gradient Descent (PGD) at clients. Wu et al.~\cite{wu2022federated} utilized class disassociation learning, client disassociation learning, and sample disassociation learning, employing reverse Stochastic Gradient Ascent (SGA) and Elastic Weight Consolidation (EWC) for joint unlearning of these three types of requests. Wu et al.~\cite{wu2022federated1} eliminated client contributions by subtracting accumulated historical updates from the model and utilized knowledge distillation methods to restore model performance without using client data. FRU~\cite{yuan2023federated} eliminates user contributions by rolling back and calibrating historical parameter updates, and then utilizes these updates to accelerate federated recommendation reconstruction. FFMU~\cite{che2023fast} utilizes nonlinear function analysis techniques to refine local machine unlearning models into output functions of Nemytskii operators, maintaining unlearning quality while improving efficiency. Wang et al.~\cite{wang2022federated} utilized CNN channel pruning to remove information about specific categories from the model for FUL processes.

\subsubsection{Re-training}

FedEraser~\cite{liu2020federated} utilized the historical parameter updates re-trained by the central server during FL to reconstruct the unlearning model. KNOT~\cite{su2023asynchronous} introduces cluster aggregation and formulates the client clustering problem as a dictionary minimization problem for re-training processes. Liu et al.~\cite{liu2022right} utilized first-order Taylor expansion approximation techniques to customize a diagonal empirical Fisher information matrix-based fast re-training algorithm.

\smallskip
Existing FUL studies predominantly assess whether the latest model version has been successfully unlearned, often overlooking crucial security challenges in FUL systems, such as trustworthiness, version control, and the traceability and accountability of unlearning iterations. They have not adequately addressed scenarios involving blockchain integration, where the need for tamper-proof data integrity and traceability introduces significant system overhead. These oversights become particularly problematic when integrating \textit{certified-by-blockchain} with \textit{model inheritance} across various prevalent FL structures, as this combination significantly amplifies the complexity and cost of FUL processes.

\subsection{Blockchain Federated Unlearning}

Islam et al.~\cite{islam2024federated} utilized redactable blockchains for flexible data updates in smart consumer electronics, enabling model update recording while allowing historical data modifications. However, they did not address the traceability and auditability of unlearning models. To enhance trustworthiness, Lin et al.~\cite{lin2024blockchain} proposed a federated unlearning proof protocol using the CH function, allowing target clients to invoke an unlearning rewrite function to erase model updates and associated data. Zuo et al.~\cite{zuo2024federated} introduced differential privacy to ensure transparency while preserving data confidentiality and reducing blockchain computational overhead. Rather than modifying on-chain data, their approach records a complete historical trace, leveraging blockchain immutability.

Different from these existing studies, our blockchain system features a dual-chain structure that supports redactability while ensuring model traceability and auditability. It also addresses real-world application requirements by enabling parallel model updates and the simultaneous update of multiple starting models. By judiciously optimizing blockchain operations, we minimize the overhead associated with redactability.

\section{Preliminary}\label{Preliminary}

\subsection{Federated Learning}
The goal of FL is to minimize the weighted average loss across $K$ clients, each with its own local dataset. The Federated Averaging (FedAvg) algorithm is typically used for this process~\cite{krauss2024automatic}.
\begin{packeditemize}
    \item \textbf{Local Training:} Each client performs several gradient descent updates on its local dataset to optimize its model.

    \item \textbf{Model Aggregation:} The server collects the local model parameters from all $K$ clients and computes their weighted average to update the global model.

    \item \textbf{Model Distribution:} The server sends the updated global model parameters to all clients for the next training round.
\end{packeditemize}

\subsection{Chameleon Hash}
The CH function allows modifying message data and its random value without changing the hash, provided the private key is known. This feature is especially useful in blockchain for scenarios requiring data adjustments~\cite{yu2020enabling}.

\begin{packeditemize}
    \item \textbf{CH Generation:} Generate a pair of public and private keys. The public key $pk$ is used for hash computation, while the private key $sk$ (trapdoor) is used for subsequent hash value modification. Given a message $m$ and a random number $r$, the CH value $h$ is computed using the CH function $CH$, resulting in
$h=CH\left( pk,m,r \right)$.

    \item \textbf{CH Verification:} For a given message $m$, random number $r$ and hash value $h$ the verification process is performed by checking whether
$h=CH\left( pk,m,r \right)$.  If this equation holds, the verification function $Verify(pk,m,r,h)$ returns $True$, confirming the validity of the hash value.

\item \textbf{CH Update:} When in possession of the private key $sk$, the hash value can be updated. Given the original message $m$, the original random number $r$, a new message $m^{\prime}$, and a new random number $r^{\prime}$, the updating process is carried out using the updating function $Update$, which computes $r^{\prime}=Update\left( sk,m,r,m^{\prime} \right)$. This ensures that the new hash value $h^{\prime}=CH\left( pk,m^{\prime} ,r^{\prime} \right)$ remains identical to the original hash value $h$, satisfying $h=h^{\prime}.$
\end{packeditemize}

\subsection{Machine Unlearning}
An unlearning process refers to scenarios where users participate in an FL task, if the training data later triggers privacy concerns or is publicly identified as damaged, a part of the training data needs to be withdrawn according to regulations, such as GDPR~\cite{voigt2017eu}. Unlearning tasks are used to remove the influence of these data from the models, ensuring that previously trained models cannot recognize such data~\cite{bourtoule2021machine}. The unlearning process is typically achieved through two primary methods: Re-training and gradient ascent.

\begin{packeditemize}
    \item \textbf{Re-training:} This straightforward approach to unlearning involves removing the specific data from the training set, re-training the model with the remaining data, and ensuring the new model excludes the unlearned information~\cite{liu2020federated}.

    \item  \textbf{Gradient ascent:} 
     This approach reverses optimization to increase the error on specific data, eliminating its influence by computing the loss, applying gradient ascent to maximize it, and iteratively reducing the dependence of the data~\cite{liu2022backdoor}.
\end{packeditemize}


\section{System Model}
\label{sec:system model}

This section presents the proposed BlockFUL framework, enabling unlearning in Blockchained FL. BlockFUL introduces a dual-chain structure—comprising a live chain and an archive chain—to facilitate controlled unlearning access. It preserves model inheritance across FL structures while ensuring efficient unlearning with comparable operational costs.

\begin{figure*}[!ht]  
\begin{center}
	\includegraphics[width=0.9\textwidth]{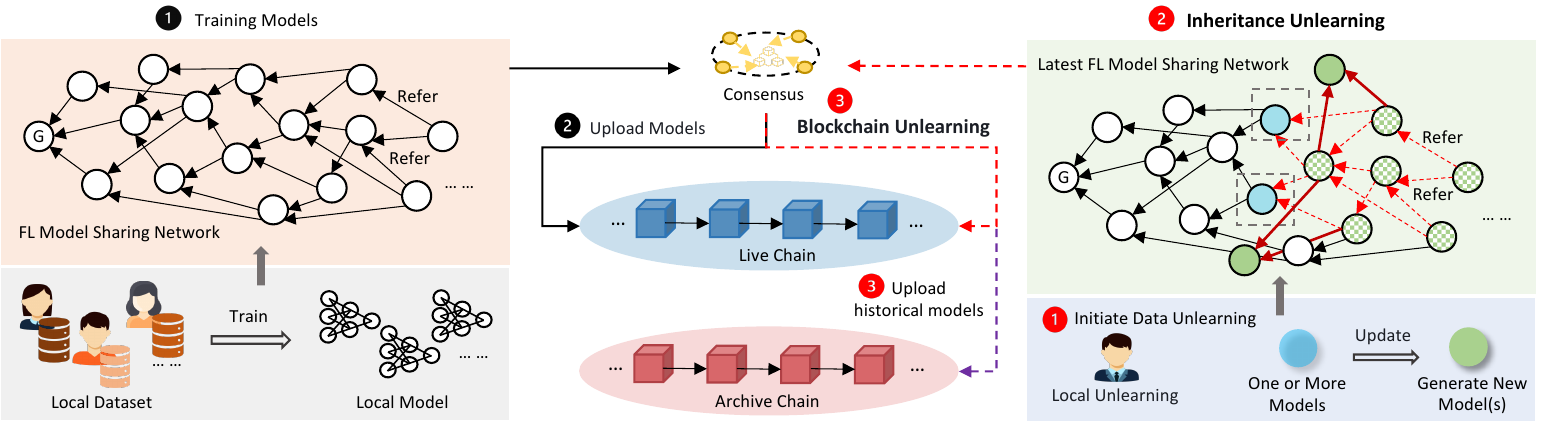}
\caption{
\small This diagram summarizes BlockFUL. The existing training stage is illustrated on the left: Users (both clients and servers) participate in an FL task, where they reference existing models to generate models and publish them to the network. The models are uploaded to a blockchain-based storage system. The proposed unlearning stage is illustrated on the right: A user needs to unlearn data (blue node, one or more models) and update models (green nodes), triggering a cascade of updates to subsequent affected models (green lattice nodes), forming a sub-DAG of the inheritance relationship. Historical model records are maintained in the ``archive chain'' of a new dual-chain structure, while the model is up-to-date in the ``live chain'' of this dual-chain structure.}

\label{fig:system framwork2}
\end{center}
\end{figure*}

\subsection{Preparation}

\subsubsection{Registration} Before the users can participate in the BlockFUL task, they need to register in the network. The users are both model publishers and model users. 

\subsubsection{Key initialization} Each user generates two key pairs, i.e., the conversation key pair ($pk, sk$) and the CH key pair ($CH_{pk}, CH_{sk}$). 

\subsubsection{Key usage and permissions} $pk$ and $sk$ are used for signing and verifying a user's legitimate identity. The CH private key $CH_{sk}$ is assigned to a committee. As a result, the committee can participate in the redactable process of the blockchain. Meanwhile, the CH public key $CH_{pk}$ is broadcast in the network. The users only have permission to share and use the model, and cannot perform change operations.

\subsection{Training Models} 

As shown in Fig. \ref{fig:system framwork2}, the BlockFUL framework is constructed based on an inheritance structure~\cite{yu2023ironforge}, where multiple users participate in FL, and each user can train multiple models. 
Consider a weighted directional graph \( G = (V, E) \), where \( V = (GV, MV) \). Herein, \( GV \) denotes the creation model node from which a task is issued, and
 $MV=\left\{ mv_{1},mv_{2},\ldots,mv_{j}\right\}$ collects model nodes in the network that participate in an FL task. 
The weighted edges, collected by $E=\left\{ e_{01},e_{10},\ldots,e_{xx}\right\}$, represent the links of inheritance relationship between the user models in the network. 

Non-existent links indicate that the weights are null. 
The inheritance relation, $r_{j\rightarrow s}=\left\{ mv_{j}\rightarrow mv_{s}\right\}$, indicates that $mv_{j}$ inherits from $mv_{s}$. Moreover, $R_{j}=\left\{ r_{j\rightarrow s},r_{j\rightarrow f},\ldots\right\}$ denotes the set of inheritance relationships referencing other model nodes from $mv_{j}$. It is further expressed as $\mathbb R_{j\rightarrow\cdots\rightarrow s}=\left\{ R_{j},\ldots\right\}$, as shown in Fig.~\ref{fig:main}. Here, $r_{j\rightarrow s}$ is one of the elements in the set of the weighted edges, $E$.

\subsubsection{Candidate models} A user employs its local test dataset $D^{test}$ to randomly select a number of models contributed by the other users for evaluation. Subsequently, the user obtains the test accuracy set $AC^{test}$ of these models until collecting $K$ candidate model sets $W^{\ast}$, as given by
\begin{equation}\label{eq.1}
\begin{split}
   (AC^{test}, W^{\ast }) &= \left\{(F(w_i, D^{test}), w_i) \, \middle| \, w_i \in K, \right. \\ & \left. \, i = 1, 2, \cdots, k\right\}, 
\end{split}
\end{equation}
where $F\left( \cdot \right)$ indicates that the user evaluates the model on its own test dataset $D^{test}$ using the model $w_i$ and obtains an accuracy. Then, this function $F(\cdot)$ returns the accuracy value. 


\subsubsection{Selection and aggregation} The user selects the top-$N$ models from the previous randomly selected $K$ candidate models, which constitute a referenced model set $\mathbb N$ for accuracy ranking. Then, the user conducts model aggregation to obtain a pre-aggregation model $\tilde{W}$, as given by

\begin{equation}\label{eq.2}
\tilde{W} =\sum_{w_{n}\in \mathbf{w}} \frac{1}{N^{R_{n}}} w_{n},
\end{equation}
where $\mathbf{w}$ is the set of models in the referenced model set $\mathbb N$ and $N^{R_{n}}$ is the number of referenced models associated with model $w_{n}$.


\subsubsection{Training} The user trains $\tilde{W}$ with its local training dataset $D^{train}$ to obtain the final aggregated model $W$:
\noindent\begin{equation}\label{eq.3}
W=\daleth \left( \tilde{W},\phi ,D^{train}\right),  
\end{equation}
where $\phi$ denotes the training settings, including the learning rate and batch size; and $\daleth \left( \cdot \right)$ denotes the user's training function for the task. 


\subsubsection{Evaluation}  After the completion of training, the user evaluates the model on its local test dataset $D^{test}$ for the final accuracy $AC$. 


\subsubsection{Node generation} The user prepares a model node $mv_{j}$ that includes the accuracy set $\mathbb A\mathbb C$ and the model set $\mathbb W$ from the referenced model set $\mathbb N$, the CH value $CH(W)$ of the model, the identifier $URI(W)$ of the model, $AC$, $\phi$ and the creation timestamp $T_{mv_{j}}$. Then, user $j$, $\forall j$, signs $mv_{j}$ and broadcasts the signed nodes in the network.

\begin{figure}[t]
\centering
\includegraphics[width=0.6\linewidth]{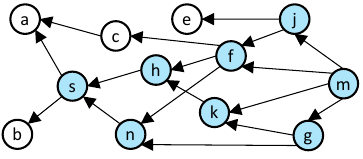}
	\caption{ \small The blue nodes are inherited nodes influenced by node $s$, namely, the child nodes of $s$. $w_{*}$ represents the model corresponding to node $*$. Between $w_{s}$ and $w_{m}$, the models affected by model $w_{s}$ are $w_{n}$, $w_{h}$, $w_{k}$, $w_{f}$, $w_{j}$ and $w_{g}$. Therefore, the set of inheritance relations  $\mathbb R_{m\rightarrow\cdots\rightarrow s}=\left\{ R_{m},R_{g},R_{j},R_{f},R_{k},R_{h},R_{n}\right\}$.}
 \label{fig:main}
\end{figure}

\subsection{Unlearning Models} 

In BlockFUL, this unlearning process needs to be conducted across the inheritance DAG, followed by an additional blockchain unlearning process.

\smallskip
\subsubsection{Inheritance unlearning phase} The unlearning request is initiated by a specific node. Employing various unlearning methods, all inherited models influenced by this node need to be updated, as shown by the green lattice nodes in Fig.~\ref{fig:system framwork2}. The updating process continues until the model nodes reach the latest generated nodes. Two cases are studied in Section~\ref{section:4}.

\smallskip
\subsubsection{Blockchain unlearning phase} As described in Section~\ref{section:3.2}, we incorporate CH into the live chain’s block structure to enable editability. Each transaction contains a CH value and a reference list $R_{Tx}$, forming a complete inherited model network and allowing updates to inherited transactions without disrupting the blockchain structure.

\change{
\subsection{Trust Model and Security Analysis}
\subsubsection{System entities}
We consider three types of participants in the system.}

\begin{packeditemize}
\item \change{\textbf{User.} Users in the FL task are assumed to be \textit{honest-but-curious}, complying with protocols and accessing models only on the live chain. Model inheritance is time-sensitive; new users cannot access models published before a specific block height, which prevents retrospective inference of early training data at the system level. BlockFUL can also employ $(\epsilon,\delta)$-differential privacy~\cite{abadi2016deep} to prevent dishonest users or attackers from inferring other users' training data through the models.}

\item\change{\textbf{Committee.} Committee members are elected blockchain miners responsible for achieving consensus on both the live chain and archive chain. A two-stage committee selection process, including pool formation and member selection, can select responsible committee members. }

\begin{enumerate}
\item\change{
In the Initial Pool Formation stage, BlockFUL can leverage a Proof-of-Stake (PoS) mechanism~\cite{gavzi2019proof}, which requires candidates to commit substantial stakes to join the committee pool. The pool formation can be augmented with a dynamic trust scoring scheme~\cite{ganeriwal2008reputation} which continuously evaluates member behavior and excludes those who fail to maintain required trust thresholds. }
\item\change{
 In the Member Selection stage, committee members are chosen from the qualified pool through a Verifiable Random Function (VRF)-based selection protocol~\cite{gilad2017algorand}, which ensures both unbiased randomness and cryptographic verifiability while maintaining complete transparency in the selection process. 
The two-stage selection process enhances system redundancy and attack resistance.
The PoS and trust score schemes work together to regulate committee behavior, ensuring committee members conduct responsible blockchained FL and FUL consensus while maintaining an up-to-date view of the live chain. Any misconduct by malicious committee members results in their exclusion from the committee pool and forfeiture of their deposited stake. }
\end{enumerate}

\item \change{\textbf{Adversary.} Adversaries are malicious entities that aim to disrupt blockchain consensus processes, including model inheritance and update verification procedures.}  
    
 \end{packeditemize}

\change{\subsubsection{Quantitative security analysis} Let $P$ denote the size of the committee pool containing $M$ malicious candidates. Using VRF, $N$ committee members are randomly selected from the pool. Given that the trust score mechanism disincentivizes malicious behavior, we model each malicious committee member as attacking with probability $\rho$. BlockFUL maintains consensus integrity as long as the number of malicious votes remains below the consensus protocol's fault tolerance threshold $f$, where $f< M\le N$. The attack success rate, denoted by $\mathcal A$, can be given by
\begin{equation}
    \mathcal A = \sum_{k=f+1}^{M} \left( \frac{{\binom{M}{k} \binom{P-M}{N-k}}}{{\binom{P}{N}}} \times \sum_{i=f+1}^{k} \binom{k}{i} \rho^i (1-\rho)^{k-i} \right)
\end{equation}
Take a Practical Byzantine Fault Tolerance (PBFT) consensus protocol tolerating up to $f=\left\lfloor \frac{N-1}{3} \right\rfloor$ malicious members for example. Setting $P=30$, $M=10$, $N=21$, and $\rho=0.2$, we obtain $f = 6$ and $\mathcal{A}=0.0000625$, demonstrating the system's robust security against potential compromises.
}

\subsection{Design Goal}
\label{sec:4.1}

Our goals for unlearning in Blockchained FL include:

\smallskip
\noindent\textbf{G1. Unlearning effectiveness.} Unlearning effectiveness is defined by the system's ability to completely eliminate the impact of unlearned data on the model. An effective unlearning process guarantees the data samples that have been unlearned no longer affect the model's predictions or parameters.

\noindent\textbf{G2. Comparable model utility.} The unlearning scheme should preserve the model accuracy on retained data categories across all updated models in the DAG of the model inheritance. This requires careful design to ensure that the unlearning process does not adversely affect the overall performance of the models and the utility of the retained data.

\noindent\textbf{G3. Support multi-start and multi-class unlearning.} 
The unlearning process can start with multiple models in the DAG. This scenario arises when a user contributes multiple models to the FL task within a period and requests unlearning from all contributed models. The unlearning tasks are expected to unlearn a single or multiple classes of data.

\noindent\textbf{G4. Manageable unlearning cost.} 
The resource and time overhead of running unlearning tasks should be reasonable and feasible. In the case of unlearning tasks with multiple starting models, the costs of unlearning with multiple starts should be substantially lower than the total cost of unlearning with individual starts.

\section{Dual-Chain Unlearning Architecture}
\label{sec:3}

In this section, we design a new dual-chain structure in the BlockFUL framework, where an archive chain is utilized to record historical models and a live chain 
shares the up-to-date FL models. This 
offers unlearning capabilities while ensuring traceability and integrity.
We design the block structure in the live chain using primitives, e.g., the CH algorithm, to provide redactable operations on the data stored on the live chain to enable the process of the blockchain unlearning stage.

\begin{figure*} [!ht]
        \centering
	\includegraphics[width=0.9\textwidth]{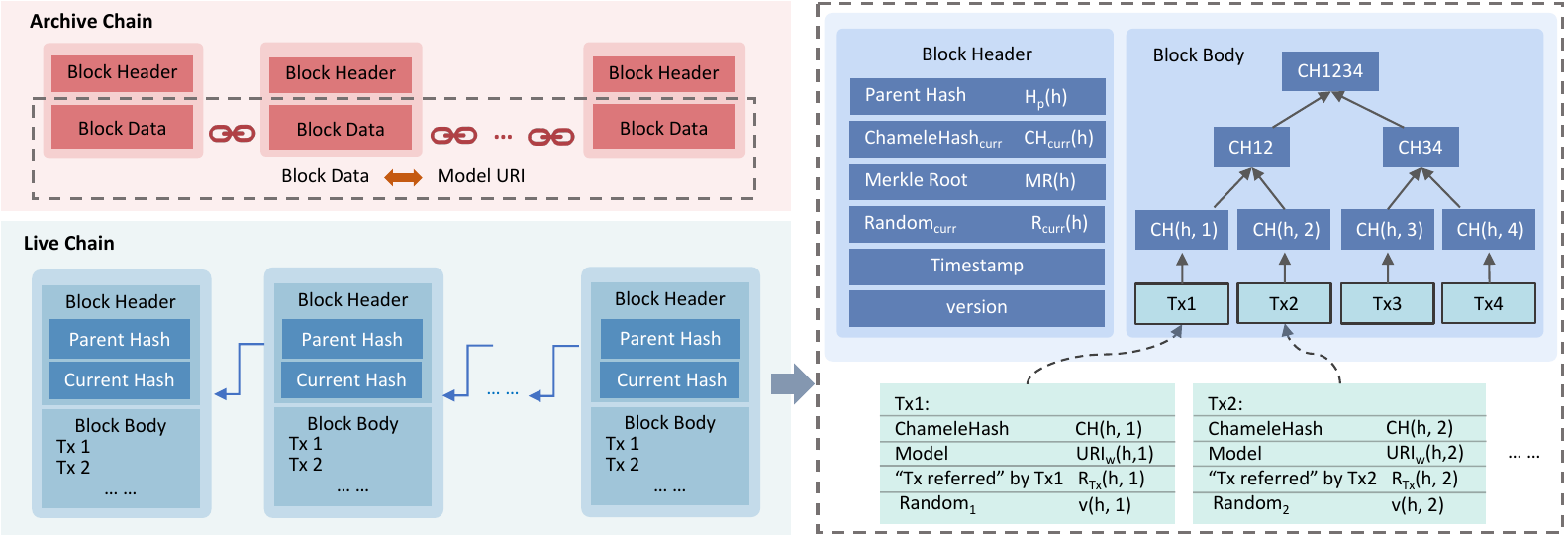}
	\caption{ \small In the dual-chain structure of BlockFUL, the archive chain stores the model's URI as block data, with multiple block data entries included within a single block. Conversely, in the live chain, the model's URI is embedded within individual transactions, with multiple transactions encapsulated in one block. Prior to any update triggered by an unlearning request, models scheduled for replacement are uploaded to the archive chain, followed by updates to the live chain that record the model in its latest version. This procedure is also applied to models undergoing multiple updates in one unlearning task.}
 \label{fig:block struture}
\end{figure*}

\subsection{Transactional Models in the Dual-Chain Structure}

This process records models in the blockchain using a dual-chain structure; see Fig.~\ref{fig:block struture}. Each transaction represents a model, with multiple transactions per block. Transactions are uploaded to archive and live chains via unified consensus, minimizing network consensus frequency. After validation, models are stored in decentralized storage, e.g., Inter Planetary File System  (IPFS), with separate instances for each chain.

\smallskip
\noindent\textbf{Archive Chain.}
The block data records generated by local participants in FL are uploaded to the archive chain. The archive chain does not grant access permission to FL users; only relevant auditing parties can view model records in the archive chain. This approach provides traceability for all model uploading activities and ensures the integrity of the models. The archive chain needs to record all historical models before and after unlearning. These historical models do not need updating, and are managed by a smart contract in a list format.

\smallskip
\noindent\textbf{Live Chain.}
The transaction records generated by local training models are uploaded to the live chain. The live chain provides model sharing and includes the latest model versions. Model updates are executed by the committee. The models are only updated on the live chain when all committee members reach a unanimous consensus. The live chain differs from the archive chain in structure. Since the live chain only records the latest state of the models, smart contracts need to operate with models as transaction entries.

In the live chain, we design the structure of the blockchain to link our blocks by having the hash of the current block's parent point to the hash of the previous block. Each block consists of a block header and a block body, consistent with the structure of traditional blocks. We use CH to ensure that transactions and block headers are redactable. This design enables BlockFUL to update multiple transactions and blocks simultaneously through a single committee consensus operation, reducing consensus frequency.\label{section:3.2}

\smallskip
\noindent\textbf{Block header.}
Fig. \ref{fig:block struture} shows a block with height $H$. The block header contains the following key fields.
\begin{packeditemize}
    \item $H_{p}\left( h\right)$: This is the parent hash of the block before linking.

    \item $CH_{curr}\left( h\right)$: Considering model updating, it functions as the hash value for the current block. This role ensures that any changes to the model version resulting from updates, do not affect the validity of the current block.

    \item $R_{curr}\left( h\right)$: The latest random number for the version-based CH value.

    \item $version$: This field is updated with the corresponding version number for each update task, maintaining consistent versions in the global state.
\end{packeditemize}
There are also the common block body's Merkle root $MR(h)$ and current block generation time $timestamp$ in block header. 

\smallskip
\noindent\textbf{Block body.} A block body contains multiple redactable transactions uniformly composed in a Merkle tree structure and stores model-related values. A redactable transaction contains the following fields, where we take transaction $Tx_{i,j}$ as an example, i.e., the $j$-th transaction in the body of block $i$.
\begin{packeditemize}
    \item CH Value $CH\left( i,j\right)$: This is the CH value of a transaction, and the field remains unchanged because it is involved in the $MR(h)$ calculation.

    \item Model $URI_{w}\left( i,j\right)$: This field stores the identifier of the model linked to the transaction. Only one model identifier is recorded per transaction. According to $URI_{w}\left( i,j\right)$, we can find the corresponding model in the IPFS.

    \item Random Number $v(i,j)$: This field stores the latest random number used to compute the CH value based on $URI_{w}\left( i,j\right)$.

    \item $Tx_{i,j}$'s reference list $R_{Tx}\left( i,j\right)$: 
    This field stores the list of transactions with models referenced by the model stored in $Tx_{i,j}$. Based on all the reference lists stored in the blockchain, we can obtain the inheritance relations of all models in the BlockFUL task.
\end{packeditemize}

The design of this block contributes to efficient and secure redactability in the blockchain. Due to the use of CH, transaction updates do not alter the block body's Merkle root. However, the version field in the block header may cause collaborative updating issues with the block header. To ensure the consistency of transaction versions, we introduce the $version$ field in the block body. As a result, the validity of the link between blocks is maintained, and computational overhead is minimized. 

Moreover, $CH\left( i,j\right)$ is a field in the block that needs verification for the block body. In this way, even if $URI_{w}\left( i,j\right)$ and $v(i, j)$ are updated, $CH(i,j)$ in the block body remains unchanged. Thus, $MR(h)$ in the block header also remains unchanged. These conditions satisfy  the immutability of $CH_{curr}\left( h\right)$ and $CH(i, j)$. A valid block with an immutable block header can still be verified.
 On the other hand, this redactability does not compromise the tamper-resistance of the blockchain. The redactability of blocks is guaranteed by $CH_{sk}$ (trapdoor), an indispensable input for CH updates~\cite{yu2020enabling}.

\subsection{Model Unlearning in Live Chain}

During the initialization phase, the CH public-private key pair is generated by the users participating in FL. We first generate a private (trapdoor) key $CH_{sk}$ and a public (hash) key $CH_{pk}$ based on a security parameter $\lambda$ and a system parameter $CH_{para}$. $CH_{pk}$ is broadcast to the network. The committee members in the blockchain each hold a $CH_{sk}$. Next, the computation involves the payload $\mathcal{T}(i,j)$ containing the hash value of the model $URI_{w}(i,j)$ and the reference list $R_{Tx}(i,j)$, CH public key $CH_{pk}$, and a random number to generate the CH value.
The calculation is as follows:
\begin{equation}\label{eq.4}
\begin{split}
&\mathcal{T}(i,j) = \{(URI_{w}(i,j), R_{Tx}(i,j)\}, \\
&CH\  Hash\left( \mathcal T(i,j),CH_{pk},v(i,j)\right)  \rightarrow ChameleHash.
\end{split}
\end{equation}

The parameter update equation related to $CH$, as specified in~(\ref{eq.5}), ensures that the original $CH$ value remains valid after information updates. The committee inputs the original payload $\mathcal{T}(i,j)$, which includes the model $URI_{w}(i,j)$ before the update, along with the new payload $\mathcal{T}^{\prime}(i,j)$ containing the updated model $URI^{\prime}_{w}(i,j)$ and the unchanged reference list $R_{Tx}(i,j)$, the random number $v(i,j)$ before the update, and the corresponding private key $CH_{sk}$. It outputs the updated random number $v^{\prime }(i,j)$, as follows:
\begin{equation}\label{eq.5}
\begin{split}
   CH\  update\left( \mathcal{T}(i,j),\mathcal{T}^{\prime}(i,j),v(i,j),CH_{sk}\right) \rightarrow 
   v^{\prime }(i,j).
  \end{split}
\end{equation}
Finally, by replacing the updated payload $\mathcal{T}^{\prime}(i,j)$ and $v^{\prime}(i,j)$ in the original $ChameleHash$, the output satisfies the following conditions:
\begin{equation}\label{eq.6}
\begin{split}
    & CH\  hash\left( \mathcal{T}^{\prime}(i,j),CH_{pk},v^{\prime }(i,j)\right) = \\
& CH\  hash\left( \mathcal{T}(i,j),CH_{pk},v(i,j)\right).
\end{split}
\end{equation}
After the transaction update is completed, the block header also needs to be updated. This is achieved by updating the version field, following a similar process to the transaction update. Once the committee completes the transaction and block updates via consensus, the information about which models have been updated and to what versions is broadcast throughout the live chain.

\section{Blockchained Federated Unlearning Paradigms}\label{section:4}

In this section, we elaborate on the new parallel and sequential unlearning paradigms in BlockFUL and explain the implementation processes by applying gradient ascent and re-training, respectively. We present the algorithms and their computational costs. Note that gradient ascent and re-training do not serve as comparative baselines, but rather as implementations for realizing the two distinct unlearning paradigms of parallelism and sequentialism.

\subsection{Paradigm 1 -- Parallel Unlearning}

\textit{Parallel unlearning} refers to the committee members collectively reaching a single consensus for updating all to-be-unlearned models within an unlearning task. After updating all models and reaching a unified consensus, the results are recorded on the chain simultaneously. To illustrate the parallel unlearning paradigm, we employ the gradient ascent to demonstrate this process.
We introduce a new gradient passing function that manipulates the ascending gradients for inherited models within the DAG using the ascending gradients from the starting models. This design ensures that updates to the starting model are efficiently propagated to all related inherited models, enabling parallel updates for single or multiple requests while reducing computational overhead.

\noindent\textbf{Update with a single starting model.} In this case, multiple paths may exist from the starting model to an inherited model. The update of the inherited model depends on the model nodes traversed by these paths and the number of models they inherit.

A user requesting to-be-unlearned data, re-trains the original model $w_{s}$ on its client to generate an updated model $w^{\prime }_{s}$. The difference between the gradients of the model before and after unlearning the data can be obtained, denoted by $\nabla \theta_{s}$. Starting from the updated model $w^{\prime }_{s}$, an inherited model $w_{y}$ of the original model $w_{s}$ could be updated accordingly as
\begin{equation}\label{eq.7}
w^{\prime }_{y}=\sum_{p_{i}\in P} \frac{\alpha \nabla \theta_{s}}{\prod_{w_j \text{ on } p_i} N^{R_j}}  +w_{y},
\end{equation}

where $P$ denotes the set of paths from $w_{s}$ to $w_{y}$, $P=\left\{ p_{1},\ldots,p_{l}\right\}$, ``$w_j \text{ on } p_i$'' denotes traversing each model node on path $p_i$ except $w_s$, $N^{R_{j}}$ denotes the number of models inherited by model node $w_j$ passing through path $p_{i}$, and $\alpha$ denotes discount factor. (\ref{eq.7}) handles model inheritance in FL, enabling the update of all relevant inherited models.

\noindent\textbf{Update with multiple starting models.} 
A user participating in a BlockFUL task contributes multiple models, and these models may have different gradients. During updates, these models serve as the starting models and have varying degrees of influence on the inherited models (i.e., the update of the inherited models may involve the computation of one or more starting model gradients). In other words, the inherited models may be affected by scenarios, e.g., the same gradient but different paths, or different gradients and paths.

For a model $w_{y}$ affected by multiple starting model updates, the update results are given by
\begin{equation}\label{eq.8}
w^{\prime }_{y}=\!\sum_{\nabla \theta^{k}_{s}\! \in \!\nabla \theta,  P^{\nabla \theta^{k}_{s}}\!\in\! \mathbb P} \sum_{p^{\nabla \theta^{k}_{s} }_{i}\in P^{\nabla \theta^{k}_{s}}} \frac{\alpha \nabla \theta^{k}_{s} }{\prod_{w_j \text{ on } p_i} N^{R_j}}\! +w_{y},
\end{equation}
where $\nabla \theta =\left\{ \nabla \theta^{1}_{s} ,\ldots ,\nabla \theta^{h}_{s} \right\}$ is the set of multiple model gradients for a user, $p_{i}^{\nabla \theta^{k}_{s} }$ is the $i$-th path with gradient $\nabla \theta^{k}_{s}$, $P^{\nabla \theta^{k}_{s}}$ denotes the set of paths with gradient $\nabla \theta^{k}_{s}$, and $\mathbb P=\left\{ P^{\nabla \theta^{1}_{s}},\ldots,P^{\nabla \theta^{h}_{s}}\right\}$ is the set of paths with all gradients.

The gradient ascent process continues on inherited models until a threshold $\varepsilon$ is reached, beyond which a model $w_y$ is not updated if $|w^{\prime}_{y}-w_{y}| \leqslant \varepsilon$. If the FL training task ends before reaching the threshold, the parallel gradient ascent process also stops.
 As shown in~(\ref{eq.7}) and~(\ref{eq.8}), the updated model $w^{\prime}_{y}$ is related to the number of model inheritances along the path from the starting models to model $w_y$. Furthermore, the number of model inheritances is associated with the depth from the starting models to model $w_y$. As the gradient ascent process carries on, a larger number of model inheritances leads to a faster reduction in the gradients and shallower updates. 

Next, we prove the boundedness of depths required for updating from the starting model to its inherited model. Let $\nabla\theta=\alpha\nabla \theta_{s}$, and consider $D=\{d_1, d_2, \ldots, d_l\}$ with $d_i$ being the depth of path $p_i$. We establish the following theorem.

\begin{theorem} \label{theorem:1}
When the number of model inheritances is $N^{R_{j}}\geqslant 2$, the depth required for model updating to complete satisfies
\begin{equation}\label{eq.9}
d_{i}\geqslant \lceil \log_{2}\frac{|\nabla \theta |}{\varepsilon } \rceil ,\ \forall d_i \in D; 
\end{equation}
\begin{equation}\label{eq.10}
\sum_{d_{i}\in D} \frac{1}{2^{d_{i}}} \leqslant \frac{\varepsilon }{|\nabla \theta| }.
\end{equation}
\end{theorem}

\begin{proof}
From (\ref{eq.7}), it readily follows that
\begin{equation}\label{eq.11}
|w^{\prime }_{y}-w_{y}|=\sum_{p_{i}\in P} \frac{1}{\prod_{w_j \text{ on } p_i} N^{R_{j}}} \times |\nabla \theta| \leqslant  \varepsilon.
\end{equation}
As discussed earlier, on a path $p_{i}$, $w_{s}$ passes through multiple model nodes to reach model $w_{y}$ with depth $d_{i}$. The number of inheritances per model is $N^{R_{j}}\geqslant 2$, with the corresponding $\frac{1}{N^{R_{j}}} \leqslant \frac{1}{2}$. For this path $p_{i}$, we have
\begin{equation}\label{eq.12}
\frac{1}{\prod_{w_j \text{ on } p_i} N^{R_{j}}}\leqslant \frac{1}{2^{d_{i}}}.
\end{equation}
Further, after passing through several model nodes, $w_{s}$  arrives at a model of depth $d_{i}$, yielding 
\begin{equation}\label{eq.13}
\frac{|\nabla\theta| }{\prod_{w_j \text{ on } p_i} N^{R_{j}}}\leqslant \frac{|\nabla\theta| }{2^{d_{i}}} \leqslant \varepsilon.
\end{equation}
For path $p_{i}$, we have
\begin{equation}\label{eq.14}
2^{d_{i}}\geqslant\frac{|\nabla\theta| }{\varepsilon} \Longrightarrow d_{i} \geqslant \lceil \log_{2}\frac{|\nabla \theta |}{\varepsilon } \rceil.
\end{equation}
For a set $P$ of all paths, we have
\begin{equation}\label{eq.15}
\begin{split}
   & \sum_{p_{i}\in P}\frac{|\nabla\theta| }{\prod_{w_j \text{ on } p_i} N^{R_{j}}}\leqslant \sum_{d_{i}\in D}\frac{|\nabla\theta| }{2^{d_{i}}}\leqslant \varepsilon 
   \Longrightarrow 
\sum_{d_{i}\in D} \frac{1}{2^{d_{i}}} \leqslant \frac{\varepsilon }{|\nabla \theta| }.
\end{split}
\end{equation}
If $\frac{1}{\prod_{w_j \text{ on } p_i} N^{R_{j}}} = \frac{1}{2^{d_{i}}}$ at this point, $d_{i}$ can be obtained as
\begin{equation}\label{eq.16}
d_{i} = 
\begin{cases} 
1, & \text{if } \log_{2}(\frac{|\nabla\theta| }{\varepsilon} )  < 0; \\
\lceil \log_{2}\frac{|\nabla \theta |}{\varepsilon } \rceil ,  & \text{if } \log_{2}(\frac{|\nabla\theta| }{\varepsilon} )  \in \mathbb{R}^+ \setminus \mathbb{Z}^+.
\end{cases}
\end{equation}
Theorem \ref{theorem:1} is proved.
\end{proof}

As dictated in Theorem \ref{theorem:1}, the minimum depth for completing a model unlearning task is determined by the variable $\nabla\theta$. Since $\nabla\theta$ is finite, there must exist a depth bound $d_{b}=\lceil \log_{2}\frac{|\nabla \theta |}{\varepsilon } \rceil $ such that models after $d_b$ do not need to be updated. When $\nabla\theta$ is smaller, it requires smaller depths for the model to stop updating. On the other hand, if $|\nabla\theta|$ is larger, it requires more depths for the model to stop updating. As mentioned earlier, the model update is related to $N^{R_{j}}$, while a bound of $N^{R_{j}}=2, \forall j$ is assumed in the proof. Therefore, when $\nabla \theta_{s}$ remains constant, the larger $N^{R_{j}}$ is, the smaller update depth bound is, and the smaller depths are required for the model to stop updating.

If there are model nodes that reference only one model for training throughout the FL training task, the depth at which the model stops updating is $d + c$. Here, $c$ represents the number of nodes inheriting only one model.

\begin{corollary} \label{corollary:1}
When multiple models of a user are used as starting points for model update, and there are interactions between these inherited models affected by the starting models, these inherited models also stop updating at a certain depth.
\end{corollary}

\begin{proof}
According to (\ref{eq.8}), assuming model $w_{y}$ is influenced by multiple distinct gradient model updates, we can represent $B^{k}_{s}$ the gradient ascent required for $w_{y}$ under the influence of gradient $\nabla \theta^{k}_{s}$. Here, $B^{k}_{s}=\sum_{p^{\nabla \theta^{k}_{s} }_{i}\in P^{\nabla \theta^{k}_{s}}} \frac{1}{\prod\nolimits{{w_j \text{ on } p_i}} N^{R_{j}}} \times \alpha \nabla \theta^{k}_{s}$. Since multiple starting models are updated simultaneously, the gradient ascent process accumulates individual gradient ascent values. In other words, $B=B^{1}_{s}+B^{2}_{s}+\cdots+B^{h}_{s}$. By Theorem \ref{theorem:1}, the inherited models influenced by the starting models eventually stop updating when a certain depth is reached. As $\prod\nolimits_{w_j \text{ on } p_i} N^{R_{j}}$ increases, $B^{k}_{s} \leqslant \varepsilon $, resulting in $B \leqslant\varepsilon$.
\end{proof}


\noindent\textbf{Parallel unlearning in BlockFUL.} Users broadcast the updated model set \( W^{\prime} \) and associated gradient messages to the network (Line 3, Algorithm~\ref{Algo.1}), after which the algorithm iterates through \( W^{\prime} \) using \textit{parallel execution loops} (Line 4). The \textit{committee} processes the updated models \( W \) and their inherited models \( N(W) \), applying \textit{gradient ascent updates} (Lines 5-8), which are then executed on the blockchain. Next, the committee utilizes the \textit{reference list} \( R_{Tx} \) to establish reference relations for all transactions in the \textit{BlockFUL} task and, using \( CH_{sk} \), overwrites relevant transactions and versions on the blockchain (Lines 9-14) until all models in \( W^{\prime} \) are processed. Finally, the committee verifies the \textit{validity of model updates} (Lines 15-21), requiring only a \textit{single consensus operation} for \textit{parallel unlearning} to update both the FL model and blockchain. A \textit{unified hash value} is computed from all prior updates, triggering an \textit{alarm} and exiting the update block if verification fails; otherwise, the update is finalized, and the new models and blockchain data are officially recorded.

\begin{algorithm}[t]
\footnotesize
\caption{Parallel unlearning} \label{Algo.1}
\SetAlgoLined

\textbf{Define} 
Receiver $\leftarrow$ Sender.Send(Message); 

\CommentSty{// This transmission is secured by any private and encrypted channels, e.g., TLS.} \\
\textbf{Input:} $G=(V,E)$, start model set $W^{\prime }$,  start model $w^{\prime }_{i} \in W^{\prime }$, model gradient set $\nabla \theta$, discount factor $\alpha$

 \For{$w^{\prime }_{i}$ \KwTo $W^{\prime }$ parallel}{
     Send model\_update message($w^{\prime }_{i}$); 

         \For{$w_{j}$ \KwTo $N(W)$}{
         $w^{\prime }_{j}$ $\leftarrow$ Committee.Update($w_{j}$, $\nabla \theta$, $\alpha$);
            
          \CommentSty{//  Update the model weights influenced with $W^{\prime}$ by (\ref{eq.8})}. \\
         $v^{\prime }_{i,j}$ $\leftarrow$  Committee.Overwrite$(\mathcal{T}(i,j)$, 
   
   $\mathcal{T}^{\prime }(i,j)$, $v_{i,j}, 
      CH_{sk}$); \\
$R^{\prime}_{curr}$ $\leftarrow$ Committee.Overwrite($version$, 
   
   $new$ $version$, $R_{curr}, CH_{sk}$); \\      
      }
       
        }

         \While{Consensus}{
           Committee.Verify$(Hash(W^{\prime }, N(W^{\prime })), v^{\prime }_{i,j}, R^{\prime}_{curr}) $\; 
    \If{$Hash(W^{\prime }, N(W^{\prime }),v^{\prime }_{i,j}, R^{\prime}_{curr})$ $is$ $invalid$}{
     alarm and exit updated blocks\; 
      update the corresponding parameters in IPFS\;
         }}
\end{algorithm}


\subsection{Paradigm 2 -- Sequential Unlearning}

\textit{Sequential unlearning} refers to a process where committee members execute a consensus process for each model update, i.e., a single DAG node, within an unlearning task. After updating one model and reaching a consensus, the results are recorded on the chain before moving on to the next model update. 
By updating models one by one and recording them on the chain, sequential unlearning ensures the compliance and procedural correctness of each operation.
We employ the re-training method to demonstrate this process, with the same user model settings adopted as in gradient ascent. The re-training starts from the updated model $w^{\prime }_{s}$ and goes through all paths to the model~$w_{y}$.

\subsubsection{Re-aggregation} Based pre-aggregation in (\ref{eq.2}), the re-aggregation model is expressed as
\begin{equation}\label{eq.17}
\tilde{w}^{\prime }_{y} =\frac{1}{N^{R_{y}}} W^{\prime }_{y},
\end{equation}
where $W^{\prime }_{y}$ denotes the updated set of the original selected model weights. 

\subsubsection{Re-training} Based on training in (\ref{eq.3}), the users associated with these paths re-train the model on their own local clients. The inherited model $w_{y}$ is updated accordingly as
\begin{equation}\label{eq.18}
w^{\prime }_{y}=\daleth_{y} \left( \tilde{w}^{\prime }_{y} ,\phi_{y} ,D^{train}_{y}\right).  
\end{equation}

\noindent\textbf{Sequential unlearning in BlockFUL.} The client sequentially applies unlearning operations to each model in \( W^{\prime} \) (Line 3, Algorithm~\ref{Algo.2}). For the current model \( w_j \), the client checks if it is the starting model. If \( w_j \neq w_{i}^{\prime} \), the client locally performs re-aggregation (\ref{eq.17}) and re-training (\ref{eq.18}) (Lines 5–8), then broadcasts the updated model to the network (Line 9). The \textit{committee}, using \( CH_{sk} \) in the blockchain, executes the same update operations as in Algorithm~\ref{Algo.1} (Lines 10–12) and initiates the consensus verification process (Lines 13–20). Upon reaching consensus, the committee updates the blocks and parameters in IPFS and forwards model update messages to the next-hop clients to update the inherited models \( N^1(W) \), \( N^2(W) \), and so on, ensuring all inherited models in \( N(W) \) are updated (Line 21).

\begin{algorithm}[t]
\footnotesize
\caption{Sequential unlearning} \label{Algo.2}
\SetAlgoLined

\textbf{Define} 
\textbf{All functions inherited from Algorithm \ref{Algo.1}}; \\
\textbf{Input:} $G=(V,E)$, start model set $W^{\prime }$,  start model $w^{\prime }_{i} \in W$\\

\For{$w^{\prime }_{i}$ \KwTo $W^{\prime }$}{
  \For{$w_{j}$ \KwTo $N(W)$}{
   \If{ $w_{j}$ $\neq$ $w^{\prime }_{i}$  }{
  Re-aggregation according to (\ref{eq.17}); \\
    Re-training according to (\ref{eq.18})\;}
 Send model\_update message($w^{\prime }_{i}$); 
     \\
      $v^{\prime }_{i,j}$ $\leftarrow$ Committee.Overwrite($\mathcal{T}(i,j)$, 
   $\mathcal{T}^{\prime }(i,j)$, $v_{i,j}, 
      CH_{sk}$); \\
$R^{\prime }_{curr}$ $\leftarrow$ Committee.Overwrite($version$, 
   
   $new$ $version$, $R_{curr}, CH_{sk}$);    
 
  \While{Consensus}{
           Committee.Verify($W^{\prime }, v^{\prime }_{i,j},R^{\prime }_{curr}$)\; 
    \If{$W^{\prime },v^{\prime }_{i,j},R^{\prime }_{curr}$ $is$ $invalid$}{
     alarm and exit \\
     update blocks\; 
      update the corresponding parameters in IPFS\;
        }
   } 
  
    \textbf{Parallel }Clients($N^{1}\left( W\right)$) $\leftarrow$ model\_update message($W^{\prime }$); 
    }
    } 

\end{algorithm}


\subsection{Analysis and Comparison}
The complexity and overhead analyses of FL and blockchain operations further highlight their computational feasibility and practical implementation.

\subsubsection{FL Complexity Analysis}

In FL networks, the computational cost of \textit{sequential updates} increases linearly with the number of requests \( Q \), leading to a complexity of \( \mathcal{O} \left( QdN^{R}|D|\times|S| \right) \), where \( d \) is the depth, \( D \) is the dataset size, and \( S \) is the number of model parameters. In contrast, \textit{parallel updates} depend primarily on depth \( d \) and remain independent of \( Q \), with a complexity of \( \mathcal{O} \left( dN^R|S| \right) \).

Considering energy consumption, assume there are \( L \) blocks and \( K \) models to be updated, including \( Q \) starting models and \( (K-Q) \) inherited models. There are $M$ nodes participating in the consensus. The energy cost for sequential updates is \( QP_{seq}\times \mathcal{O} \left( dN^R|D|\times |S| \right) \), while for parallel updates, it is given by \( QP_{seq}\times \mathcal{O} \left( |D|\times |S| \right)+(K-Q)MP_{para}\times \mathcal{O} \left( dN^R|S| \right) \), where \( P_{seq} \) and \( P_{para} \) are the computational power for sequential and parallel updates, respectively.

\subsubsection{Blockchain Complexity Analysis}
The update is executed every time a node is traversed, regardless of whether or not unlearning has been performed on the node. 
The number of consensus processes varies based on the unlearning techniques used: a single consensus occurs in parallel unlearning, whereas multiple consensuses are required in sequential unlearning. The computational overheads of the unlearning paradigms on the archive chain are unaffected by model updates.

The transmission cost $C_{tran}$ represents the cost associated with a single operation of uploading $C^{up}_{tran}$ or downloading $C^{down}_{tran}$ a model. The committee consensus cost is $C_{con}$. The cost of a single CH update operation is denoted by $C_{CH}$.

\smallskip
 The analysis of parallel unlearning costs is as follows:

\noindent\textbf{CH cost.} The cost of CH includes the cost of transactions CH $C^{Tx}_{CH}$, and the cost of blocks CH $C^{block}_{CH}$ in the blockchain. Therefore, the CH cost of parallel unlearning is $C^{Pa}_{CH}=K C^{Tx}_{CH}+N C^{block}_{CH}=(K+L) C_{CH}$.

\noindent\textbf{Block consensus cost.} The committee executes the entire parallel unlearning process, enabling consensus to be reached for multiple transactions or blocks simultaneously. The consensus cost of parallel unlearning is $C^{Pa}_{con}=C_{con}$.

\noindent\textbf{Transmission cost.} The transmission cost of the model is primarily incurred at the consensus nodes for uploading and downloading. Since the consensus nodes can perform uploading and downloading tasks in parallel, the transmission cost of the model is $C^{Pa}_{tran}=K C^{up}_{tran}+K C^{down}_{tran}=2K C_{tran}$.

For parallel operations, the blockchain overhead is 
\begin{equation}\label{eq.parallel consumption}
(K+L) C_{CH}+C_{con}+2K C_{tran},
\end{equation}
which can also derive the estimated energy consumption: 
\begin{equation}
    \left( K+L \right)M E_{CH}+ME_{con}+2K|S| E_{tran},
\end{equation}
where $E_{CH}$ is the CH energy consumption; $E_{con}$ represents the energy consumption per node for consensus among $M$ participants, $E_{tran}$ is the transmission energy consumption per unit of data size, and $S$ is the number of model parameters.

\smallskip
The analysis of sequential unlearning costs is as follows:
 
\noindent\textbf{CH cost.} In this process, each transaction update requires one transaction CH and one block CH. The CH cost of sequential unlearning is $C^{Se}_{CH}=K C^{Tx}_{CH}+K C^{block}_{CH}=2K C_{CH}$.

\noindent\textbf{Block consensus cost.} Each model update requires one round of consensus. The consensus cost for sequential unlearning is $C^{Se}_{con}=K C_{con}$.

\noindent\textbf{Transmission cost.} The transmission cost of the model primarily involves the uploading and downloading processes by local clients, as well as the downloading process by consensus nodes. Suppose that the average number of inheritances per model is $N^{R}$. The transmission cost for sequential unlearning is $C^{Se}_{tran}=K C^{up}_{tran}+N^{R} (K-1)C^{down}_{tran}+K C^{down}_{tran}=2K C_{tran}+N^{R} (K-1) C_{tran}$.

For sequential operations, the blockchain overhead is
\begin{equation}\label{eq.parallel consumption}
2K C_{CH}+K C_{con}+2K C_{tran}+N^R (K-1) C_{tran},
\end{equation}
which can also derive the estimated energy consumption: 
\begin{equation}\label{eq.parallel consumption}
2KME_{CH}+KME_{con}+2K|S| E_{tran}+N^{R}\left( K-1 \right) |S| E_{tran}.
\end{equation}

\smallskip
Additionally, we assess other metrics, including \textit{blockchain overhead} in cases where each transaction forms an individual block (\( SC_{block} \)) and where multiple transactions are grouped within a single block (\( MC_{block} \)). We also examine the impact of varying training dataset sizes (\( |D| \)), considering factors such as depth (\( d \)) and parameter count (\( |S| \)).

One or more transaction updates go through CH overhead $C_{CH}$. Similarly, the main overhead is the same as the transaction overhead for one block or multiple block updates. Suppose that there are $K$ model updates. The costs can be given by $SC_{block}=K (C^{block}_{CH}+C^{Tx}_{CH})=2K C_{CH}$. In the case of multiple transactions corresponding to a single block, suppose there are $K$ model updates, then the update cost in this block is $MC_{block}=C^{block}_{CH}+K C^{Tx}_{CH}=(K+1) C_{CH}$.

In large-scale datasets and complex models, computational cost limitations manifest in several ways. First, as dataset size $|D|$ and model complexity (e.g., depth $d$ and parameter count $|S|$) increase, computational costs rise significantly, especially in sequential updates, where complexity grows linearly with the number of requests $Q$, reducing system efficiency. Second, in blockchain systems, the consensus overhead $\mathcal{O} \left( C_{con} \right)$ and on-chain processing overhead $\mathcal{O} \left( C_{CH} \right)$ further add to the computational burden. Finally, while parallel updates can reduce computational complexity, they rely on hardware parallelization, which may not be fully exploited in resource-constrained environments.

\subsubsection{Comparison}

The design of parallel and sequential unlearning paradigms aligns with real-world applications, where the choice of updating entity significantly impacts computational costs. In financial prediction models, parallel unlearning facilitates rapid updates from multiple market data sources, enabling real-time trend monitoring and dynamic portfolio adjustments to respond effectively to market fluctuations. Conversely, in medical compliance and record management, unlearning is infrequent but critical, often triggered by data lifecycle expiration or regulatory changes, requiring strict compliance and structured execution, making sequential unlearning the preferred approach.

The method of executing updates also affects efficiency. If individual users perform updates, each must upload their modified model to the blockchain and notify the next user, necessitating separate consensus processes and incurring significant transmission overhead. As the number of updated models surpasses the number of block updates, the cumulative cost of CH in the blocks exceeds that of collective updates managed by a central committee. Delegating updates to the committee reduces computational overhead on clients, optimizing consensus and transmission costs as model numbers grow. This allows the committee to maintain a more efficient balance between system performance and cost management, ensuring scalable and streamlined unlearning processes.

\section{Experiments}\label{section:5}

In the experiment, we establish a testbed to assess the effectiveness of
parallel and sequential unlearning, and evaluate their impact 
on accuracy for both unlearned and retained labels, incorporating classic models, such as ResNet18 and MobileNetV2. We also evaluate the overhead of interacting with the dual-chain structure.

We align the gradient ascent-based parallel unlearning and re-training-based sequential unlearning. This alignment is reasonable since each method aptly reflects the characteristics of its respective unlearning paradigm: re-training, which requires sequentially fetching and training each model, and gradient ascent, which facilitates parallel processing due to its one-off fetching capability from the live chain\footnote{Other unlearning methods, such as fine-tuning and distillation, can also be implemented within these paradigms. For the sake of clarity and focus, we select representative methods for our evaluation in this paper.}.

\subsection{Evaluation Metrics}
We focus on class-level unlearning, though BlockFUL supports class-, client-, and sample-level unlearning using methods like re-training and gradient ascent~\cite{liu2020federated,liu2022backdoor}. We evaluate accuracy on training data to check if removed information still affects the model, following standard practice~\cite{lastlayer}.

\begin{packeditemize}
    \item\textbf{Accuracy on the unlearned dataset ($AD_f$).} 
The accuracy of an unlearned dataset in the unlearned model ideally is close to zero. This aligns with \textbf{G1} by evaluating the success of unlearning through the model’s inability to predict labels corresponding to the unlearned dataset accurately.

  \item\textbf{Accuracy on the retained dataset ($AD_r$).} The accuracy of a retained dataset of the unlearned model. It is expected to be close to the performance of the original model. This aligns with \textbf{G2} by maintaining comparable accuracy on retained labels corresponding to the retained dataset.

  \item\textbf{Cumulative unlearning time.} The cumulative time required for each model to unlearn labels during the training process in machine learning and deep learning.
   \item\textbf{Blockchain overhead.} This includes the time overhead
required for CH and consensus in the unlearning process
on the blockchain, as well as the time overhead for
processing single or multiple transactions within a block.

   \item\textbf{Energy consumption.} This mainly includes the energy consumption of the FL network model updates and the energy consumption of the blockchain network.

  \item\textbf{Latency.} This includes the cumulative unlearning time and the time spent on blockchain and CH operations.

\end{packeditemize}

\begin{table*}
\centering
\setlength{\extrarowheight}{0pt}
\addtolength{\extrarowheight}{\aboverulesep}
\addtolength{\extrarowheight}{\belowrulesep}
\setlength{\aboverulesep}{0pt}
\setlength{\belowrulesep}{0pt}
\caption{\small Unlearning Performance on CIFAR-10}
\label{table:1}
\renewcommand\arraystretch{0.6}
\resizebox{1\textwidth}{!}{
\begin{tabular}{ccc|ccccccccc|cccccc} 
\toprule
\multirow{3}{*}{\textbf{Model}}        & \multirow{3}{*}{\textbf{\#$C_f$}} & \multirow{3}{*}{\textbf{Metrics}}                 & \multicolumn{3}{c|}{\multirow{2}{*}{Original~(\%)}}                                                                   & \multicolumn{3}{c|}{\multirow{2}{*}{Re-training~(\%)}}                                                             & \multicolumn{3}{c|}{\multirow{2}{*}{Gradient Ascent~(\%)}}                                                         & \multicolumn{6}{c}{\textbf{Cumulative Unlearning Time (s)}}                                                                                                \\ 
\cmidrule{13-18}
                                       &                                   &                                                   & \multicolumn{3}{c|}{}                                                                                                 & \multicolumn{3}{c|}{}                                                                                              & \multicolumn{3}{c|}{}                                                                                              & \multicolumn{3}{c|}{Re-training}                                               & \multicolumn{3}{c}{Gradient Ascent}                                       \\ 
\cmidrule{4-18}
                                       &                                   &                                                   & $r_5$                                 & $r_6$                                 & $r_8$                                 & $r_5$                                & $r_6$                                & $r_8$                                & $r_5$                                & $r_6$                                & $r_8$                                & $r_5$                   & $r_6$                   & \multicolumn{1}{c|}{$r_8$} & $r_5$                  & $r_6$                  & $r_8$                   \\
\bottomrule
\end{tabular}
}
\end{table*}
\begin{table*}
\centering
\setlength{\extrarowheight}{0pt}
\addtolength{\extrarowheight}{\aboverulesep}
\addtolength{\extrarowheight}{\belowrulesep}
\setlength{\aboverulesep}{0pt}
\setlength{\belowrulesep}{0pt}
\caption{\small Unlearning Performance on Fashion-MNIST}
\label{table:overall}
\renewcommand\arraystretch{0.6}
    \resizebox{1\textwidth}{!}{
\begin{tabular}{ccc|ccccccccc|cccccc} 
\toprule
\multirow{3}{*}{\textbf{Model}}     & \multirow{3}{*}{\textbf{\#$C_f$}} & \multirow{3}{*}{\textbf{Metrics}}                 & \multicolumn{3}{c|}{\multirow{2}{*}{Original~(\%)}}                                                                   & \multicolumn{3}{c|}{\multirow{2}{*}{Re-training~(\%)}}                                                             & \multicolumn{3}{c|}{\multirow{2}{*}{Gradient Ascent~(\%)}}                                                         & \multicolumn{6}{c}{\textbf{Cumulative Unlearning Time (s)}}                                                                                                \\ 
\cmidrule{13-18}
                                    &                                   &                                                   & \multicolumn{3}{c|}{}                                                                                                 & \multicolumn{3}{c|}{}                                                                                              & \multicolumn{3}{c|}{}                                                                                              & \multicolumn{3}{c|}{Re-training}                                               & \multicolumn{3}{c}{Gradient Ascent}                                       \\ 
\cmidrule{4-18}
                                    &                                   &                                                   & $r_5$                                 & $r_6$                                 & $r_8$                                 & $r_5$                                & $r_6$                                & $r_8$                                & $r_5$                                & $r_6$                                & $r_8$                                & $r_5$                   & $r_6$                   & \multicolumn{1}{c|}{$r_8$} & $r_5$                  & $r_6$                  & $r_8$                   \\
\bottomrule
\end{tabular}
}
\end{table*}
\begin{table*}
\centering
\setlength{\extrarowheight}{0pt}
\addtolength{\extrarowheight}{\aboverulesep}
\addtolength{\extrarowheight}{\belowrulesep}
\setlength{\aboverulesep}{0pt}
\setlength{\belowrulesep}{0pt}
\caption{\small Non-IID and IID Unlearning Performance on CIFAR-100 (ResNet18)}
\label{table:IID-NON}
\renewcommand\arraystretch{0.6}
\resizebox{1\textwidth}{!}{
\begin{tabular}{ccc|ccccccccc|cccccc} 
\toprule
\multirow{3}{*}{\textbf{Methods}}         & \multirow{3}{*}{\textbf{\#$C_f$}} & \multirow{3}{*}{\textbf{Metrics}}                 & \multicolumn{3}{c|}{\multirow{2}{*}{Original~(\%)}}                                                                   & \multicolumn{3}{c|}{\multirow{2}{*}{Re-training~(\%)}}                                                             & \multicolumn{3}{c|}{\multirow{2}{*}{Gradient Ascent~(\%)}}                                                           & \multicolumn{6}{c}{\textbf{Cumulative Unlearning Time (s)}}                                                                                       \\ 
\cmidrule{13-18}
                                          &                                   &                                                   & \multicolumn{3}{c|}{}                                                                                                 & \multicolumn{3}{c|}{}                                                                                              & \multicolumn{3}{c|}{}                                                                                                & \multicolumn{3}{c|}{Re-training}                                         & \multicolumn{3}{c}{Gradient Ascent}                                    \\ 
\cmidrule{4-18}
                                          &                                   &                                                   & $r_5$                                 & $r_6$                                 & $r_8$                                 & $r_5$                                & $r_6$                                & $r_8$                                & $r_5$                                 & $r_6$                                 & $r_8$                                & $r_5$                  & $r_6$                  & $r_8$                  & $r_5$                 & $r_6$                 & $r_8$                  \\
\bottomrule
\end{tabular}
}
\end{table*}
\begin{table*}
\centering
\setlength{\extrarowheight}{0pt}
\addtolength{\extrarowheight}{\aboverulesep}
\addtolength{\extrarowheight}{\belowrulesep}
\setlength{\aboverulesep}{0pt}
\setlength{\belowrulesep}{0pt}
\caption{\small Non-IID and IID Unlearning Performance on Yahoo! Answers (Bert-Base-Cased)}
\label{table:IID-NON-bert}
\renewcommand\arraystretch{0.6}
\resizebox{1\textwidth}{!}{
\begin{tabular}{ccc|ccccccccc|cccccc} 
\toprule
\multirow{3}{*}{\textbf{Methods}}         & \multirow{3}{*}{\textbf{\#$C_f$}} & \multirow{3}{*}{\textbf{Metrics}}                 & \multicolumn{3}{c|}{\multirow{2}{*}{Original~(\%)}}                                                                   & \multicolumn{3}{c|}{\multirow{2}{*}{Re-training~(\%)}}                                                             & \multicolumn{3}{c|}{\multirow{2}{*}{Gradient Ascent~(\%)}}                                                            & \multicolumn{6}{c}{\textbf{Cumulative Unlearning Time (s)}}                                                                                          \\ 
\cmidrule{13-18}
                                          &                                   &                                                   & \multicolumn{3}{c|}{}                                                                                                 & \multicolumn{3}{c|}{}                                                                                              & \multicolumn{3}{c|}{}                                                                                                 & \multicolumn{3}{c|}{Re-training}                                            & \multicolumn{3}{c}{Gradient Ascent}                                    \\ 
\cmidrule{4-18}
                                          &                                   &                                                   & $r_5$                                 & $r_6$                                 & $r_8$                                 & $r_5$                                & $r_6$                                & $r_8$                                & $r_5$                                 & $r_6$                                 & $r_8$                                 & $r_5$                   & $r_6$                   & $r_8$                   & $r_5$                 & $r_6$                 & $r_8$                  \\
\bottomrule
\end{tabular}
}
\end{table*}
\begin{table*}
\centering
\setlength{\extrarowheight}{0pt}
\addtolength{\extrarowheight}{\aboverulesep}
\addtolength{\extrarowheight}{\belowrulesep}
\setlength{\aboverulesep}{0pt}
\setlength{\belowrulesep}{0pt}
\caption{\small Non-IID Unlearning Performance on CIFAR-100 (ResNet18)}
\label{table:IID-NON-cifar-compare}\
\renewcommand\arraystretch{0.6}
\resizebox{1\textwidth}{!}{
\begin{tabular}{ccc|ccccccccc|cccccc} 
\toprule
\multirow{3}{*}{\textbf{Methods}}           & \multirow{3}{*}{\textbf{\#$C_f$}} & \multirow{3}{*}{\textbf{Metrics}}                 & \multicolumn{3}{c|}{\multirow{2}{*}{Original~(\%)}}                                                                   & \multicolumn{3}{c|}{\multirow{2}{*}{Re-training~(\%)}}                                                             & \multicolumn{3}{c|}{\multirow{2}{*}{Gradient Ascent~(\%)}}                                                         & \multicolumn{6}{c}{\textbf{Cumulative Unlearning Time (s)}}                                                                                       \\ 
\cmidrule{13-18}
                                            &                                   &                                                   & \multicolumn{3}{c|}{}                                                                                                 & \multicolumn{3}{c|}{}                                                                                              & \multicolumn{3}{c|}{}                                                                                              & \multicolumn{3}{c|}{Re-training}                                         & \multicolumn{3}{c}{Gradient Ascent}                                    \\ 
\cmidrule{4-18}
                                            &                                   &                                                   & $r_5$                                 & $r_6$                                 & $r_8$                                 & $r_5$                                & $r_6$                                & $r_8$                                & $r_5$                                & $r_6$                                & $r_8$                                & $r_5$                  & $r_6$                  & $r_8$                  & $r_5$                 & $r_6$                 & $r_8$                  \\
\bottomrule
\end{tabular}
}
\end{table*}
\begin{table*}
\centering
\setlength{\extrarowheight}{0pt}
\addtolength{\extrarowheight}{\aboverulesep}
\addtolength{\extrarowheight}{\belowrulesep}
\setlength{\aboverulesep}{0pt}
\setlength{\belowrulesep}{0pt}
\caption{\small Non-IID Unlearning Performance on Yahoo! Answers (Bert-Base-Cased)}
\label{table:IID-NON-bert-compare}
  \renewcommand\arraystretch{0.6}
    \resizebox{1\textwidth}{!}{
\begin{tabular}{ccc|ccccccccc|cccccc} 
\toprule
\multirow{3}{*}{\textbf{Methods}}          & \multirow{3}{*}{\textbf{\#$C_f$}} & \multirow{3}{*}{\textbf{Metrics}}                 & \multicolumn{3}{c|}{\multirow{2}{*}{Original~(\%)}}                                                                   & \multicolumn{3}{c|}{\multirow{2}{*}{Re-training~(\%)}}                                                             & \multicolumn{3}{c|}{\multirow{2}{*}{Gradient Ascent~(\%)}}                                                          & \multicolumn{6}{c}{\textbf{Cumulative Unlearning Time (s)}}                                                                                          \\ 
\cmidrule{13-18}
                                           &                                   &                                                   & \multicolumn{3}{c|}{}                                                                                                 & \multicolumn{3}{c|}{}                                                                                              & \multicolumn{3}{c|}{}                                                                                               & \multicolumn{3}{c|}{Re-training}                                            & \multicolumn{3}{c}{Gradient Ascent}                                    \\ 
\cmidrule{4-18}
                                           &                                   &                                                   & $r_5$                                 & $r_6$                                 & $r_8$                                 & $r_5$                                & $r_6$                                & $r_8$                                & $r_5$                                 & $r_6$                                & $r_8$                                & $r_5$                   & $r_6$                   & $r_8$                   & $r_5$                 & $r_6$                 & $r_8$                  \\
\bottomrule
\end{tabular}
}
\end{table*}

\subsection{Datasets}

\begin{packeditemize}
    \item \textbf{CIFAR-10:} 
    Comprises 60,000 32x32 color images in 10 different classes, with 6,000 images per class. The dataset is split into 50,000 training and 10,000 testing images. 

    \item \textbf{Fashion-MNIST:}  
    Comprises 10 fashion categories, each comprising 60,000 grayscale images with a dimension of 28x28 pixels. The training set contains 55,000 images, while the test set includes 10,000 images. 

    \item \textbf{Yahoo! Answers:}  Contains approximately 1.4 million entries across 10 categories, offering a valuable resource for natural language processing (NLP) tasks. 

    \item \textbf{TinyImageNet:} The dataset contains 200 classes, with 500 images per class in the training set and an additional 10,000 images in the test set.
\end{packeditemize}

\subsection{Implementation of Unlearning Paradigms}

Each traversal triggers an ``unlearning'' at every node, regardless of whether it has been unlearned previously or not. We employ the following unlearning methods:

\smallskip
\noindent\textbf{Re-training (Sequential).}
The model is re-trained sequentially on the remaining data after removing the unlearned data, following the inheritance path.

\smallskip
\noindent\textbf{Gradient ascent (Parallel).}
We calculate gradient differences from initial unlearning and apply them to inherited models in parallel, enabling efficient unlearning and analysis of its downstream effects.

\begin{figure}[t]
\centering
\includegraphics[width=0.8\linewidth]{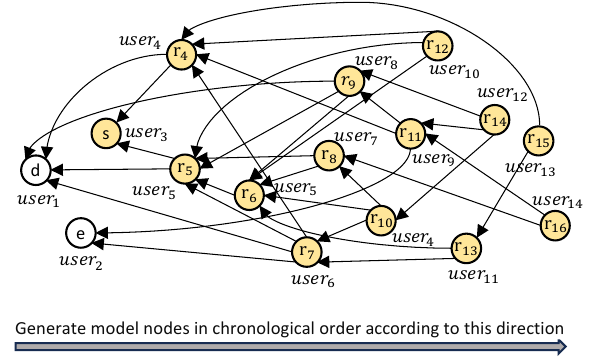}
	\caption{\small The experiment considers 14 users, and each arrow represents the aggregated model. The yellow color indicates unlearned data models and their inheritance models.}
 \label{unlearninggraph}
\end{figure}

\subsection{Experimental Settings}
We conduct experiments using an NVIDIA GeForce RTX 3070 GPU (Tables~\ref{table:1} and~\ref{table:overall}) and an RTX 4090 GPU (Tables~\ref{table:IID-NON}-\ref{table:fedrated_TinyImageNet} and Section~\ref{Scalability}). We employ an untrained ResNet18 model, MobileNetV2 model, Bert-Base-Cased model, GCNN model, and AlexNet model, all fine-tuned on specific datasets. The to-be-unlearned classes during the process are represented as \#$C_f$.

\begin{packeditemize}
    \item\noindent \textbf{Model Inheritance:}
A group of 14 users is evaluated, as depicted in Fig. \ref{unlearninggraph}, where $user_5$ and $user_7$ have the same unlearning labels. 
The models highlighted in yellow are those that demonstrate accuracy in unlearned data.
   
\item\noindent \textbf{Blockchain:} In the live chain, we use Byzantine Fault Tolerance (BFT)-based consensus to ensure data consistency and security. Additionally, we set the block packaging time to 15 seconds. This interval strikes a balance between transaction processing speed and network stability. Each block is sized at 1MB to ensure efficient network transmission and storage. 

\item\noindent \textbf{Hyperparameters:}  In both CIFAR-10 and Fashion-MNIST setups, each model undergoes rigorous training for 100 epochs with a learning rate of 0.005. 

\item\noindent \textbf{Label categories: }  We apply the unlearning process to both single and multiple category labels (such as combinations of 2, 4, etc., category labels). Here, $r_5$, $r_6$, and $r_8$ represent the models randomly targeted for the unlearning process.
\end{packeditemize}

\subsection{Evaluating Unlearning Performance}

We evaluate re-training and gradient ascent by comparing \( AD_f \) and \( AD_r \), corresponding to design goals \textbf{G1} and \textbf{G2}, while assessing single- and multi-label unlearning to address \textbf{G3}. Tables~\ref{table:1} and~\ref{table:overall} report results on CIFAR-10 and Fashion-MNIST using AlexNet, which shows shorter re-training time and better \( AD_r \) than ResNet18, indicating that smaller models may unlearn more effectively via re-training.

\subsubsection{Unlearning effectiveness of single-label unlearning}

On CIFAR-10 and Fashion-MNIST, re-training consistently achieves perfect unlearning ($AD_f$=0) with high retention ($AD_r$ up to 92.77\% and 95.27\%, respectively). In contrast, gradient ascent yields moderate unlearning ($AD_f$ around 7\% and 3\%) but lower retention ($AD_r$of 51.09\% and 83.4\%). In single-label unlearning, gradient ascent matches re-training in $AD_f$, effectively forgetting the target label. However, its lower $AD_r$ indicates weaker retention of other labels.

\subsubsection{Unlearning effectiveness of multi-label unlearning} When \#$C_f $ (i.e., the number of label categories to be unlearned) is set to 2, 4, and 7 across all models on CIFAR-10, the $AD_f$ metric for gradient ascent is 6.64\%, and the average $AD_r$ is 48.77\%. On Fashion-MNIST, the experimental results show that the $AD_f$ for gradient ascent is 1.64\%, and the average $AD_r$ is 61.1\%. Gradient ascent performs similarly to re-training on $AD_f$, but its declining $AD_r$ suggests that focusing on unlearning labels hampers retention of other labels, affecting overall performance.

Different models and datasets balance computational cost and accuracy differently. CIFAR-10, with its larger dataset, demands more re-training resources than Fashion-MNIST but achieves higher accuracy. Model complexity also matters—ResNet18, being more complex than MobileNetV2, requires higher re-training costs but delivers superior accuracy. In gradient ascent, ResNet18 maintains its accuracy advantage over MobileNetV2 despite similar computational costs, highlighting the benefits of model complexity in this method. Thus, model and dataset selection should match task requirements: MobileNetV2 suits resource-constrained scenarios (e.g., real-time applications), while ResNet18 is ideal for accuracy-critical tasks (e.g., image recognition).

\vspace{-0.4em}
\begin{center}
\fbox{%
\begin{minipage}{0.98\linewidth}
\change{\textbf{Takeaway—balancing unlearning accuracy and scalability.}
Re-training offers perfect forgetting with
high retention, particularly for lightweight models like AlexNet. Gradient ascent, while less effective
in retention, enables cost-efficient unlearning with acceptable forgetting performance, and benefits
from model complexity. These results support flexible method selection based on task constraints
such as accuracy demand and computational resources.
}
\end{minipage}
}
\end{center}

\change{
\subsubsection{Non-IID and IID} Tables~\ref{table:IID-NON} and~\ref{table:IID-NON-bert} present experimental results under IID and non-IID data distributions, where client data is either sample-non-IID or label-non-IID. 
In the sample-non-IID scenario, clients share the same classes but with different data within each class; in the label-non-IID scenario, clients have disjoint class distributions (e.g., client 1 does not have classes 70-74, and client~2 does not have classes 60-64). For IID data, clients 1 and 2 share both labels and data.
It is observed that the GA method effectively achieves unlearning in most cases, with particularly better performance under the sample-non-IID setting. As \#$C_f$ increases, GA accuracy under IID data drops significantly because sample-non-IID maintains greater data diversity, enhancing model performance, whereas label-non-IID’s disjoint class distribution can negatively impact learning.
}

\change{Tables~\ref{table:IID-NON-cifar-compare} and~\ref{table:IID-NON-bert-compare} present the experimental results of non-IID unlearning performance under different label settings on CIFAR-100 (ResNet18) and Yahoo! Answers (Bert-Base-Cased), respectively. We observe that as the number of label-non-IID labels increases (e.g., in the CIFAR-100 dataset (ResNet18), where the number of inconsistent labels increases from 5 to 15), the unlearning performance remains relatively high across different models and datasets.
}

Beyond re-training and gradient ascent, we explored distillation and fine-tuning. As shown in Tables~\ref{table:2} and~\ref{table:11}, distillation slightly underperforms re-training in accuracy while incurring higher time overhead, and fine-tuning proves ineffective. Fine-tuning updates only the final layers of a pre-trained model on a dataset excluding the unlearned data. To improve unlearning and minimize $AD_f$, we applied a large learning rate, effectively removing unlearned data but causing instability due to distribution shifts, making it harder to retain essential information~\cite{li2017learning}.

\begin{table*}
\centering
\setlength{\extrarowheight}{0pt}
\addtolength{\extrarowheight}{\aboverulesep}
\addtolength{\extrarowheight}{\belowrulesep}
\setlength{\aboverulesep}{0pt}
\setlength{\belowrulesep}{0pt}
\caption{\small Fine-tuning and Distill Unlearning Performance on CIFAR-100}
\label{table:2}
    \renewcommand\arraystretch{0.6}
\resizebox{1\textwidth}{!}{
\begin{tabular}{ccc|ccccccccc|cccccc} 
\toprule
\multirow{3}{*}{\textbf{Model}}     & \multirow{3}{*}{\textbf{\#$C_f$}} & \multirow{3}{*}{\textbf{Metrics}}                 & \multicolumn{3}{c|}{\multirow{2}{*}{Original (\%)}}                                                                   & \multicolumn{3}{c|}{\multirow{2}{*}{Fine-tuning (\%)}}                                                              & \multicolumn{3}{c|}{\multirow{2}{*}{Distill (\%)}}                                                                 & \multicolumn{6}{c}{\textbf{Cumulative Unlearning Time (s)}}                                                                                              \\ 
\cmidrule{13-18}
                                    &                                   &                                                   & \multicolumn{3}{c|}{}                                                                                                 & \multicolumn{3}{c|}{}                                                                                               & \multicolumn{3}{c|}{}                                                                                              & \multicolumn{3}{c|}{Fine-tuning}                                           & \multicolumn{3}{c}{Distill}                                                 \\ 
\cmidrule{4-18}
                                    &                                   &                                                   & $r_5$                                 & $r_6$                                 & $r_8$                                 & $r_5$                                & $r_6$                                & $r_8$                                 & $r_5$                                & $r_6$                                & $r_8$                                & $r_5$                 & $r_6$                 & \multicolumn{1}{c|}{$r_8$} & $r_5$                  & $r_6$                   & $r_8$                    \\ 
\midrule
\multirow{4}{*}{\rotcell{AlexNet}}  & \multirow{2}{*}{1}                & $A_{D_r}\uparrow$                                   & 99.99                                 & 99.99                                 & 99.99                                 & 1.03                                 & 1.03                                 & 3.21                                  & 81.74                                & 86.25                                & 86.76                                & \multirow{2}{*}{0.85} & \multirow{2}{*}{1.16} & \multirow{2}{*}{2.45}      & \multirow{2}{*}{83.14} & \multirow{2}{*}{154.61} & \multirow{2}{*}{296.74}  \\
                                    &                                   & {\cellcolor[rgb]{0.922,0.922,1}}$A_{D_f}\downarrow$ & {\cellcolor[rgb]{0.922,0.922,1}}99.99 & {\cellcolor[rgb]{0.922,0.922,1}}99.99 & {\cellcolor[rgb]{0.922,0.922,1}}99.99 & {\cellcolor[rgb]{0.922,0.922,1}}0.00 & {\cellcolor[rgb]{0.922,0.922,1}}0.00 & {\cellcolor[rgb]{0.922,0.922,1}}0.00  & {\cellcolor[rgb]{0.922,0.922,1}}0.00 & {\cellcolor[rgb]{0.922,0.922,1}}0.00 & {\cellcolor[rgb]{0.922,0.922,1}}0.00 &                       &                       &                            &                        &                         &                          \\ 
\cline{2-18}
                                    & \multirow{2}{*}{10}               & $A_{D_r}\uparrow$~                                  & 99.99                                 & 99.99                                 & 99.99                                 & 1.13                                 & 4.77                                 & 67.04                                 & 80.44                                & 81.07                                & 91.32                                & \multirow{2}{*}{0.84} & \multirow{2}{*}{2.14} & \multirow{2}{*}{3.31}      & \multirow{2}{*}{76.31} & \multirow{2}{*}{149.60} & \multirow{2}{*}{329.95}  \\
                                    &                                   & {\cellcolor[rgb]{0.922,0.922,1}}$A_{D_f}\downarrow$ & {\cellcolor[rgb]{0.922,0.922,1}}99.99 & {\cellcolor[rgb]{0.922,0.922,1}}99.99 & {\cellcolor[rgb]{0.922,0.922,1}}99.99 & {\cellcolor[rgb]{0.922,0.922,1}}0.00 & {\cellcolor[rgb]{0.922,0.922,1}}0.00 & {\cellcolor[rgb]{0.922,0.922,1}}52.73 & {\cellcolor[rgb]{0.922,0.922,1}}0.00 & {\cellcolor[rgb]{0.922,0.922,1}}0.00 & {\cellcolor[rgb]{0.922,0.922,1}}0.00 &                       &                       &                            &                        &                         &                          \\ 
\midrule
\multirow{4}{*}{\rotcell{ResNet18}} & \multirow{2}{*}{1}                & $A_{D_r}\uparrow$                                   & 99.99                                 & 99.99                                 & 99.99                                 & 1.03                                 & 2.44                                 & 1.03                                  & 84.93                                & 84.93                                & 81.32                                & \multirow{2}{*}{1.04} & \multirow{2}{*}{2.32} & \multirow{2}{*}{3.74}      & \multirow{2}{*}{71.31} & \multirow{2}{*}{131.44} & \multirow{2}{*}{250.61}  \\
                                    &                                   & {\cellcolor[rgb]{1,0.992,0.859}}$A_{D_f}\downarrow$ & {\cellcolor[rgb]{1,0.992,0.859}}99.99 & {\cellcolor[rgb]{1,0.992,0.859}}99.99 & {\cellcolor[rgb]{1,0.992,0.859}}99.99 & {\cellcolor[rgb]{1,0.992,0.859}}0.00 & {\cellcolor[rgb]{1,0.992,0.859}}0.00 & {\cellcolor[rgb]{1,0.992,0.859}}0.00  & {\cellcolor[rgb]{1,0.992,0.859}}0.00 & {\cellcolor[rgb]{1,0.992,0.859}}0.00 & {\cellcolor[rgb]{1,0.992,0.859}}0.00 &                       &                       &                            &                        &                         &                          \\ 
\cline{2-18}
                                    & \multirow{2}{*}{10}               & $A_{D_r}\uparrow$                                   & 99.99                                 & 99.99                                 & 99.99                                 & 6.87                                 & 1.68                                 & 2.04                                  & 81.02                                & 79.76                                & 91.04                                & \multirow{2}{*}{1.85} & \multirow{2}{*}{2.54} & \multirow{2}{*}{3.84}      & \multirow{2}{*}{76.12} & \multirow{2}{*}{152.33} & \multirow{2}{*}{270.34}  \\
                                    &                                   & {\cellcolor[rgb]{1,0.992,0.859}}$A_{D_f}\downarrow$ & {\cellcolor[rgb]{1,0.992,0.859}}99.98 & {\cellcolor[rgb]{1,0.992,0.859}}99.99 & {\cellcolor[rgb]{1,0.992,0.859}}99.99 & {\cellcolor[rgb]{1,0.992,0.859}}2.61 & {\cellcolor[rgb]{1,0.992,0.859}}0.00 & {\cellcolor[rgb]{1,0.992,0.859}}0.00  & {\cellcolor[rgb]{1,0.992,0.859}}0.00 & {\cellcolor[rgb]{1,0.992,0.859}}0.00 & {\cellcolor[rgb]{1,0.992,0.859}}0.00 &                       &                       &                            &                        &                         &                          \\
\bottomrule
\end{tabular}
}
\end{table*}
\begin{table*}
\centering
\setlength{\extrarowheight}{0pt}
\addtolength{\extrarowheight}{\aboverulesep}
\addtolength{\extrarowheight}{\belowrulesep}
\setlength{\aboverulesep}{0pt}
\setlength{\belowrulesep}{0pt}
\caption{\small Re-training and Gradient Ascent Unlearning Performance on CIFAR-100}
\label{table:11}
\renewcommand\arraystretch{0.6}
\resizebox{1\textwidth}{!}{
\begin{tabular}{ccc|ccccccccc|cccccc} 
\toprule
\multirow{3}{*}{\textbf{Model}} & \multirow{3}{*}{\textbf{\#$C_f$}} & \multirow{3}{*}{\textbf{Metrics}}                 & \multicolumn{3}{c|}{\multirow{2}{*}{Original (\%)}}                                                                   & \multicolumn{3}{c|}{\multirow{2}{*}{Re-training (\%)}}                                                             & \multicolumn{3}{c|}{\multirow{2}{*}{Gradient Ascent (\%)}}                                                            & \multicolumn{6}{c}{\textbf{Cumulative Unlearning Time (s)}}                                                                                           \\ 
\cmidrule{13-18}
                                &                                   &                                                   & \multicolumn{3}{c|}{}                                                                                                 & \multicolumn{3}{c|}{}                                                                                              & \multicolumn{3}{c|}{}                                                                                                 & \multicolumn{3}{c|}{Re-training}                                             & \multicolumn{3}{c}{Gradient Ascent}                                    \\ 
\cmidrule{4-18}
                                &                                   &                                                   & $r_5$                                 & $r_6$                                 & $r_8$                                 & $r_5$                                & $r_6$                                & $r_8$                                & $r_5$                                 & $r_6$                                 & $r_8$                                 & $r_5$                  & $r_6$                  & \multicolumn{1}{c|}{$r_8$} & $r_5$                 & $r_6$                 & $r_8$                  \\ 
\midrule
\multirow{4}{*}{AlexNet}        & \multirow{2}{*}{1}                & $A_{D_r}\uparrow$                                   & 99.99                                 & 99.99                                 & 99.99                                 & 99.99                                & 99.99                                & 99.99                                & 99.13                                 & 92.69                                 & 96.32                                 & \multirow{2}{*}{30.14} & \multirow{2}{*}{59.53} & \multirow{2}{*}{91.31}     & \multirow{2}{*}{1.34} & \multirow{2}{*}{1.85} & \multirow{2}{*}{2.42}  \\
                                &                                   & {\cellcolor[rgb]{0.922,0.922,1}}$A_{D_f}\downarrow$ & {\cellcolor[rgb]{0.922,0.922,1}}99.99 & {\cellcolor[rgb]{0.922,0.922,1}}99.99 & {\cellcolor[rgb]{0.922,0.922,1}}99.99 & {\cellcolor[rgb]{0.922,0.922,1}}0.00 & {\cellcolor[rgb]{0.922,0.922,1}}0.00 & {\cellcolor[rgb]{0.922,0.922,1}}0.00 & {\cellcolor[rgb]{0.922,0.922,1}}1.38  & {\cellcolor[rgb]{0.922,0.922,1}}0.00  & {\cellcolor[rgb]{0.922,0.922,1}}0.00  &                        &                        &                            &                       &                       &                        \\ 
\cline{2-18}
                                & \multirow{2}{*}{10}               & $A_{D_r}\uparrow$~                                  & 99.99                                 & 99.99                                 & 99.99                                 & 99.99                                & 99.99                                & 99.99                                & 43.61                                 & 39.30                                 & 49.94                                 & \multirow{2}{*}{29.53} & \multirow{2}{*}{57.16} & \multirow{2}{*}{89.17}     & \multirow{2}{*}{1.67} & \multirow{2}{*}{2.14} & \multirow{2}{*}{2.75}  \\
                                &                                   & {\cellcolor[rgb]{0.922,0.922,1}}$A_{D_f}\downarrow$ & {\cellcolor[rgb]{0.922,0.922,1}}99.99 & {\cellcolor[rgb]{0.922,0.922,1}}99.99 & {\cellcolor[rgb]{0.922,0.922,1}}99.99 & {\cellcolor[rgb]{0.922,0.922,1}}0.00 & {\cellcolor[rgb]{0.922,0.922,1}}0.00 & {\cellcolor[rgb]{0.922,0.922,1}}0.00 & {\cellcolor[rgb]{0.922,0.922,1}}21.05 & {\cellcolor[rgb]{0.922,0.922,1}}17.76 & {\cellcolor[rgb]{0.922,0.922,1}}29.31 &                        &                        &                            &                       &                       &                        \\ 
\midrule
\multirow{4}{*}{ResNet18}       & \multirow{2}{*}{1}                & $A_{D_r}\uparrow$                                   & 99.99                                 & 99.99                                 & 99.99                                 & 99.99                                & 99.99                                & 99.99                                & 92.81                                 & 79.55                                 & 87.06                                 & \multirow{2}{*}{31.64} & \multirow{2}{*}{62.51} & \multirow{2}{*}{92.64}     & \multirow{2}{*}{1.84} & \multirow{2}{*}{2.35} & \multirow{2}{*}{2.94}  \\
                                &                                   & {\cellcolor[rgb]{1,0.992,0.859}}$A_{D_f}\downarrow$ & {\cellcolor[rgb]{1,0.992,0.859}}99.99 & {\cellcolor[rgb]{1,0.992,0.859}}99.99 & {\cellcolor[rgb]{1,0.992,0.859}}99.99 & {\cellcolor[rgb]{1,0.992,0.859}}0.00 & {\cellcolor[rgb]{1,0.992,0.859}}0.00 & {\cellcolor[rgb]{1,0.992,0.859}}0.00 & {\cellcolor[rgb]{1,0.992,0.859}}1.38  & {\cellcolor[rgb]{1,0.992,0.859}}0.00  & {\cellcolor[rgb]{1,0.992,0.859}}0.00  &                        &                        &                            &                       &                       &                        \\ 
\cline{2-18}
                                & \multirow{2}{*}{10}               & $A_{D_r}\uparrow$                                   & 99.99                                 & 99.99                                 & 99.99                                 & 99.99                                & 99.99                                & 99.99                                & 81.02                                 & 79.76                                 & 91.04                                 & \multirow{2}{*}{29.95} & \multirow{2}{*}{59.14} & \multirow{2}{*}{88.18}     & \multirow{2}{*}{1.93} & \multirow{2}{*}{2.45} & \multirow{2}{*}{3.01}  \\
                                &                                   & {\cellcolor[rgb]{1,0.992,0.859}}$A_{D_f}\downarrow$ & {\cellcolor[rgb]{1,0.992,0.859}}99.98 & {\cellcolor[rgb]{1,0.992,0.859}}99.99 & {\cellcolor[rgb]{1,0.992,0.859}}99.99 & {\cellcolor[rgb]{1,0.992,0.859}}0.00 & {\cellcolor[rgb]{1,0.992,0.859}}0.00 & {\cellcolor[rgb]{1,0.992,0.859}}0.00 & {\cellcolor[rgb]{1,0.992,0.859}}0.00  & {\cellcolor[rgb]{1,0.992,0.859}}0.00  & {\cellcolor[rgb]{1,0.992,0.859}}0.00  &                        &                        &                            &                       &                       &                        \\
\bottomrule
\end{tabular}
}
\end{table*}
\begin{table*}
\centering
\setlength{\extrarowheight}{0pt}
\addtolength{\extrarowheight}{\aboverulesep}
\addtolength{\extrarowheight}{\belowrulesep}
\setlength{\aboverulesep}{0pt}
\setlength{\belowrulesep}{0pt}
\caption{\small Large-scale Federated Unlearning Performance of Language and Graph Models on Yahoo! Answers}
\label{table:Yahoo}
\renewcommand\arraystretch{0.6}
\resizebox{1\textwidth}{!}{
\begin{tabular}{ccc|ccccccccc|cccccc} 
\toprule
\multirow{3}{*}{\textbf{Model}}  & \multirow{3}{*}{\textbf{\#$C_f$}} & \multirow{3}{*}{\textbf{Metrics}}                 & \multicolumn{3}{c|}{\multirow{2}{*}{Original (\%)}}                                                                   & \multicolumn{3}{c|}{\multirow{2}{*}{Re-training~(\%)}}                                                             & \multicolumn{3}{c|}{\multirow{2}{*}{Gradient Ascent~(\%)}}                                                            & \multicolumn{6}{c}{\textbf{Cumulative Unlearning Time (s)}}                                                                                             \\ 
\cmidrule{13-18}
                                 &                                   &                                                   & \multicolumn{3}{c|}{}                                                                                                 & \multicolumn{3}{c|}{}                                                                                              & \multicolumn{3}{c|}{}                                                                                                 & \multicolumn{3}{c|}{Re-training}                                               & \multicolumn{3}{c}{Gradient Ascent}                                    \\ 
\cmidrule{4-18}
                                 &                                   &                                                   & $r_5$                                 & $r_6$                                 & $r_8$                                 & $r_5$                                & $r_6$                                & $r_8$                                & $r_5$                                 & $r_6$                                 & $r_8$                                 & $r_5$                   & $r_6$                   & \multicolumn{1}{c|}{$r_8$} & $r_5$                 & $r_6$                 & $r_8$                  \\ 
\midrule
\multirow{6}{*}{Bert-Base-Cased} & \multirow{2}{*}{1}                & $A_{D_r}\uparrow$                                   & 99.99                                 & 99.99                                 & 99.99                                 & 99.99                                & 99.99                                & 99.99                                & 84.40                                 & 96.39                                 & 84.49                                 & \multirow{2}{*}{105.63} & \multirow{2}{*}{215.14} & \multirow{2}{*}{313.75}    & \multirow{2}{*}{1.97} & \multirow{2}{*}{3.05} & \multirow{2}{*}{4.14}  \\
                                 &                                   & {\cellcolor[rgb]{0.922,0.922,1}}$A_{D_f}\downarrow$ & {\cellcolor[rgb]{0.922,0.922,1}}99.99 & {\cellcolor[rgb]{0.922,0.922,1}}99.99 & {\cellcolor[rgb]{0.922,0.922,1}}99.99 & {\cellcolor[rgb]{0.922,0.922,1}}0.00 & {\cellcolor[rgb]{0.922,0.922,1}}0.00 & {\cellcolor[rgb]{0.922,0.922,1}}0.00 & {\cellcolor[rgb]{0.922,0.922,1}}3.96  & {\cellcolor[rgb]{0.922,0.922,1}}8.91  & {\cellcolor[rgb]{0.922,0.922,1}}6.25  &                         &                         &                            &                       &                       &                        \\ 
\cline{2-18}
                                 & \multirow{2}{*}{2}                & $A_{D_r}\uparrow$~                                  & 99.99                                 & 99.99                                 & 99.99                                 & 99.99                                & 99.99                                & 99.99                                & 71.48                                 & 43.74                                 & 30.87                                 & \multirow{2}{*}{95.61}  & \multirow{2}{*}{190.38} & \multirow{2}{*}{308.74}    & \multirow{2}{*}{2.07} & \multirow{2}{*}{3.51} & \multirow{2}{*}{4.81}  \\
                                 &                                   & {\cellcolor[rgb]{0.922,0.922,1}}$A_{D_f}\downarrow$ & {\cellcolor[rgb]{0.922,0.922,1}}99.99 & {\cellcolor[rgb]{0.922,0.922,1}}99.99 & {\cellcolor[rgb]{0.922,0.922,1}}99.99 & {\cellcolor[rgb]{0.922,0.922,1}}0.00 & {\cellcolor[rgb]{0.922,0.922,1}}0.00 & {\cellcolor[rgb]{0.922,0.922,1}}0.00 & {\cellcolor[rgb]{0.922,0.922,1}}15.13 & {\cellcolor[rgb]{0.922,0.922,1}}19.45 & {\cellcolor[rgb]{0.922,0.922,1}}14.51 &                         &                         &                            &                       &                       &                        \\ 
\cline{2-18}
                                 & \multirow{2}{*}{4}                & $A_{D_r}\uparrow$                                   & 99.99                                 & 99.99                                 & 99.99                                 & 99.99                                & 99.99                                & 99.99                                & 70.47                                 & 49.31                                 & 50.85                                 & \multirow{2}{*}{72.16}  & \multirow{2}{*}{147.47} & \multirow{2}{*}{232.35}    & \multirow{2}{*}{2.57} & \multirow{2}{*}{4.34} & \multirow{2}{*}{5.42}  \\
                                 &                                   & {\cellcolor[rgb]{0.922,0.922,1}}$A_{D_f}\downarrow$ & {\cellcolor[rgb]{0.922,0.922,1}}99.99 & {\cellcolor[rgb]{0.922,0.922,1}}99.99 & {\cellcolor[rgb]{0.922,0.922,1}}99.99 & {\cellcolor[rgb]{0.922,0.922,1}}0.00 & {\cellcolor[rgb]{0.922,0.922,1}}0.00 & {\cellcolor[rgb]{0.922,0.922,1}}0.00 & {\cellcolor[rgb]{0.922,0.922,1}}12.29 & {\cellcolor[rgb]{0.922,0.922,1}}13.63 & {\cellcolor[rgb]{0.922,0.922,1}}18.18 &                         &                         &                            &                       &                       &                        \\ 
\midrule
\multirow{6}{*}{GCNN}            & \multirow{2}{*}{1}                & $A_{D_r}\uparrow$                                   & 99.56                                 & 99.93                                 & 96.43                                 & 99.69                                & 99.99                                & 97.64                                & 58.42                                 & 44.06                                 & 42.48                                 & \multirow{2}{*}{35.14}  & \multirow{2}{*}{71.36}  & \multirow{2}{*}{107.25}    & \multirow{2}{*}{1.02} & \multirow{2}{*}{2.02} & \multirow{2}{*}{3.17}  \\
                                 &                                   & {\cellcolor[rgb]{1,0.992,0.859}}$A_{D_f}\downarrow$ & {\cellcolor[rgb]{1,0.992,0.859}}98.81 & {\cellcolor[rgb]{1,0.992,0.859}}99.99 & {\cellcolor[rgb]{1,0.992,0.859}}94.64 & {\cellcolor[rgb]{1,0.992,0.859}}0.00 & {\cellcolor[rgb]{1,0.992,0.859}}0.00 & {\cellcolor[rgb]{1,0.992,0.859}}0.00 & {\cellcolor[rgb]{1,0.992,0.859}}14.79 & {\cellcolor[rgb]{1,0.992,0.859}}8.33  & {\cellcolor[rgb]{1,0.992,0.859}}20.71 &                         &                         &                            &                       &                       &                        \\ 
\cline{2-18}
                                 & \multirow{2}{*}{2}                & $A_{D_r}\uparrow$                                   & 99.56                                 & 99.93                                 & 96.43                                 & 99.99                                & 99.99                                & 98.31                                & 60.57                                 & 41.48                                 & 40.27                                 & \multirow{2}{*}{30.66}  & \multirow{2}{*}{61.24}  & \multirow{2}{*}{92.34}     & \multirow{2}{*}{1.44} & \multirow{2}{*}{2.87} & \multirow{2}{*}{3.86}  \\
                                 &                                   & {\cellcolor[rgb]{1,0.992,0.859}}$A_{D_f}\downarrow$ & {\cellcolor[rgb]{1,0.992,0.859}}99.11 & {\cellcolor[rgb]{1,0.992,0.859}}99.99 & {\cellcolor[rgb]{1,0.992,0.859}}95.60 & {\cellcolor[rgb]{1,0.992,0.859}}0.00 & {\cellcolor[rgb]{1,0.992,0.859}}0.00 & {\cellcolor[rgb]{1,0.992,0.859}}0.00 & {\cellcolor[rgb]{1,0.992,0.859}}17.80 & {\cellcolor[rgb]{1,0.992,0.859}}8.01  & {\cellcolor[rgb]{1,0.992,0.859}}13.19 &                         &                         &                            &                       &                       &                        \\ 
\cline{2-18}
                                 & \multirow{2}{*}{4}                & $A_{D_r}\uparrow$                                   & 99.56                                 & 99.93                                 & 96.43                                 & 99.99                                & 99.99                                & 98.34                                & 54.34                                 & 35.80                                 & 41.72                                 & \multirow{2}{*}{23.31}  & \multirow{2}{*}{46.46}  & \multirow{2}{*}{71.87}     & \multirow{2}{*}{2.04} & \multirow{2}{*}{3.14} & \multirow{2}{*}{4.51}  \\
                                 &                                   & {\cellcolor[rgb]{1,0.992,0.859}}$A_{D_f}\downarrow$ & {\cellcolor[rgb]{1,0.992,0.859}}99.07 & {\cellcolor[rgb]{1,0.992,0.859}}99.99 & {\cellcolor[rgb]{1,0.992,0.859}}95.67 & {\cellcolor[rgb]{1,0.992,0.859}}0.00 & {\cellcolor[rgb]{1,0.992,0.859}}0.00 & {\cellcolor[rgb]{1,0.992,0.859}}0.00 & {\cellcolor[rgb]{1,0.992,0.859}}10.54 & {\cellcolor[rgb]{1,0.992,0.859}}10.00 & {\cellcolor[rgb]{1,0.992,0.859}}11.19 &                         &                         &                            &                       &                       &                        \\
\bottomrule
\end{tabular}
}
\end{table*}
\begin{table*}
\centering
\setlength{\extrarowheight}{0pt}
\addtolength{\extrarowheight}{\aboverulesep}
\addtolength{\extrarowheight}{\belowrulesep}
\setlength{\aboverulesep}{0pt}
\setlength{\belowrulesep}{0pt}
\caption{\small Federated Unlearning Performance on TinyImageNet}
\label{table:fedrated_TinyImageNet}
\renewcommand\arraystretch{0.6}
\resizebox{1\textwidth}{!}{
\begin{tabular}{ccc|ccccccccc|cccccc} 
\toprule
\multirow{3}{*}{\textbf{Model}} & \multirow{3}{*}{\textbf{\#$C_f$}} & \multirow{3}{*}{\textbf{Metrics}}                 & \multicolumn{3}{c|}{\multirow{2}{*}{Original (\%)}}                                                                   & \multicolumn{3}{c|}{\multirow{2}{*}{Re-training~(\%)}}                                                             & \multicolumn{3}{c|}{\multirow{2}{*}{Gradient Ascent (\%)}}                                                            & \multicolumn{6}{c}{\textbf{Cumulative Unlearning Time (s)}}                                                                                                \\ 
\cmidrule{13-18}
                                &                                   &                                                   & \multicolumn{3}{c|}{}                                                                                                 & \multicolumn{3}{c|}{}                                                                                              & \multicolumn{3}{c|}{}                                                                                                 & \multicolumn{3}{c|}{Re-training}                                               & \multicolumn{3}{c}{Gradient Ascent}                                       \\ 
\cmidrule{4-18}
                                &                                   &                                                   & $r_5$                                 & $r_6$                                 & $r_8$                                 & $r_5$                                & $r_6$                                & $r_8$                                & $r_5$                                 & $r_6$                                 & $r_8$                                 & $r_5$                   & $r_6$                   & \multicolumn{1}{c|}{$r_8$} & $r_5$                  & $r_6$                  & $r_8$                   \\ 
\midrule
\multirow{8}{*}{ResNet18}       & \multirow{2}{*}{1}                & $A_{D_r}\uparrow$                                   & 99.99                                 & 99.99                                 & 99.99                                 & 99.99                                & 99.99                                & 99.99                                & 99.98                                 & 99.98                                 & 99.99                                 & \multirow{2}{*}{114.95} & \multirow{2}{*}{220.39} & \multirow{2}{*}{320.29}    & \multirow{2}{*}{21.44} & \multirow{2}{*}{25.78} & \multirow{2}{*}{28.62}  \\
                                &                                   & {\cellcolor[rgb]{0.922,0.922,1}}$A_{D_f}\downarrow$ & {\cellcolor[rgb]{0.922,0.922,1}}99.99 & {\cellcolor[rgb]{0.922,0.922,1}}99.99 & {\cellcolor[rgb]{0.922,0.922,1}}99.99 & {\cellcolor[rgb]{0.922,0.922,1}}0.00 & {\cellcolor[rgb]{0.922,0.922,1}}0.00 & {\cellcolor[rgb]{0.922,0.922,1}}0.00 & {\cellcolor[rgb]{0.922,0.922,1}}0.00  & {\cellcolor[rgb]{0.922,0.922,1}}1.44  & {\cellcolor[rgb]{0.922,0.922,1}}18.60 &                         &                         &                            &                        &                        &                         \\ 
\cline{2-18}
                                & \multirow{2}{*}{4}                & $A_{D_r}\uparrow$                                   & 99.99                                 & 99.99                                 & 99.99                                 & 99.99                                & 99.99                                & 99.99                                & 83.25                                 & 93.73                                 & 81.06                                 & \multirow{2}{*}{111.89} & \multirow{2}{*}{213.31} & \multirow{2}{*}{316.42}    & \multirow{2}{*}{20.67} & \multirow{2}{*}{24.97} & \multirow{2}{*}{27.51}  \\
                                &                                   & {\cellcolor[rgb]{0.922,0.922,1}}$A_{D_f}\downarrow$ & {\cellcolor[rgb]{0.922,0.922,1}}99.99 & {\cellcolor[rgb]{0.922,0.922,1}}99.99 & {\cellcolor[rgb]{0.922,0.922,1}}99.99 & {\cellcolor[rgb]{0.922,0.922,1}}0.00 & {\cellcolor[rgb]{0.922,0.922,1}}0.00 & {\cellcolor[rgb]{0.922,0.922,1}}0.00 & {\cellcolor[rgb]{0.922,0.922,1}}7.16  & {\cellcolor[rgb]{0.922,0.922,1}}21.49 & {\cellcolor[rgb]{0.922,0.922,1}}14.33 &                         &                         &                            &                        &                        &                         \\ 
\cline{2-18}
                                & \multirow{2}{*}{6}                & $A_{D_r}\uparrow$                                   & 99.99                                 & 99.99                                 & 99.99                                 & 99.96                                & 99.99                                & 99.99                                & 78.36                                 & 59.99                                 & 68.41                                 & \multirow{2}{*}{109.39} & \multirow{2}{*}{212.12} & \multirow{2}{*}{311.32}    & \multirow{2}{*}{20.13} & \multirow{2}{*}{24.56} & \multirow{2}{*}{26.47}  \\
                                &                                   & {\cellcolor[rgb]{1,0.992,0.859}}$A_{D_f}\downarrow$ & {\cellcolor[rgb]{1,0.992,0.859}}99.99 & {\cellcolor[rgb]{1,0.992,0.859}}99.99 & {\cellcolor[rgb]{1,0.992,0.859}}99.99 & {\cellcolor[rgb]{1,0.992,0.859}}0.00 & {\cellcolor[rgb]{1,0.992,0.859}}0.00 & {\cellcolor[rgb]{1,0.992,0.859}}0.00 & {\cellcolor[rgb]{1,0.992,0.859}}12.94 & {\cellcolor[rgb]{1,0.992,0.859}}11.27 & {\cellcolor[rgb]{1,0.992,0.859}}19.34 &                         &                         &                            &                        &                        &                         \\ 
\cline{2-18}
                                & \multirow{2}{*}{10}               & $A_{D_r}\uparrow$                                   & 99.99                                 & 99.99                                 & 99.99                                 & 99.99                                & 99.99                                & 99.99                                & 57.39                                 & 65.26                                 & 56.41                                 & \multirow{2}{*}{106.21} & \multirow{2}{*}{205.36} & \multirow{2}{*}{308.74}    & \multirow{2}{*}{18.30} & \multirow{2}{*}{21.67} & \multirow{2}{*}{23.68}  \\
                                &                                   & {\cellcolor[rgb]{1,0.992,0.859}}$A_{D_f}\downarrow$ & {\cellcolor[rgb]{1,0.992,0.859}}99.99 & {\cellcolor[rgb]{1,0.992,0.859}}99.99 & {\cellcolor[rgb]{1,0.992,0.859}}99.99 & {\cellcolor[rgb]{1,0.992,0.859}}0.00 & {\cellcolor[rgb]{1,0.992,0.859}}0.00 & {\cellcolor[rgb]{1,0.992,0.859}}0.00 & {\cellcolor[rgb]{1,0.992,0.859}}12.81 & {\cellcolor[rgb]{1,0.992,0.859}}13.30 & {\cellcolor[rgb]{1,0.992,0.859}}12.80 &                         &                         &                            &                        &                        &                         \\
\bottomrule
\end{tabular}
}
\end{table*}

Table~\ref{table:Yahoo} presents multi-label unlearning results on the Yahoo! Answers dataset using the Bert-Base-Cased and GCNN models with 50 clients. As client numbers increase, $AD_f$ slightly rises (11.19\% to 20.71\%) but remains low overall. BlockFUL achieves better time efficiency in GA than re-training, despite a minor trade-off in $AD_r$, with GA significantly reducing computational costs. Similar trends appear in Table~\ref{table:fedrated_TinyImageNet} for TinyImageNet, where GA consistently outperforms re-training. The $AD_r$ values remain within an acceptable range, ensuring the method’s effectiveness.

\vspace{-0.4em}
\begin{center}
\fbox{%
\begin{minipage}{0.98\linewidth}
\change{\textbf{Takeaway—scalable unlearning under non-IID and federated settings.}
Gradient ascent maintains
strong unlearning performance across IID and non-IID scenarios, particularly excelling under sample-
non-IID due to preserved data diversity. It remains effective as the number of forgotten labels increases
and scales well across federated setups, offering significant time savings over re-training with only
minor trade-offs in retention accuracy.
}
\end{minipage}
}
\end{center}

\subsection{Evaluating Unlearning Scalability}\label{Scalability}

\begin{figure*}[ht]
    \centering
    \subfigure[]{
        \includegraphics[width=0.21\textwidth]{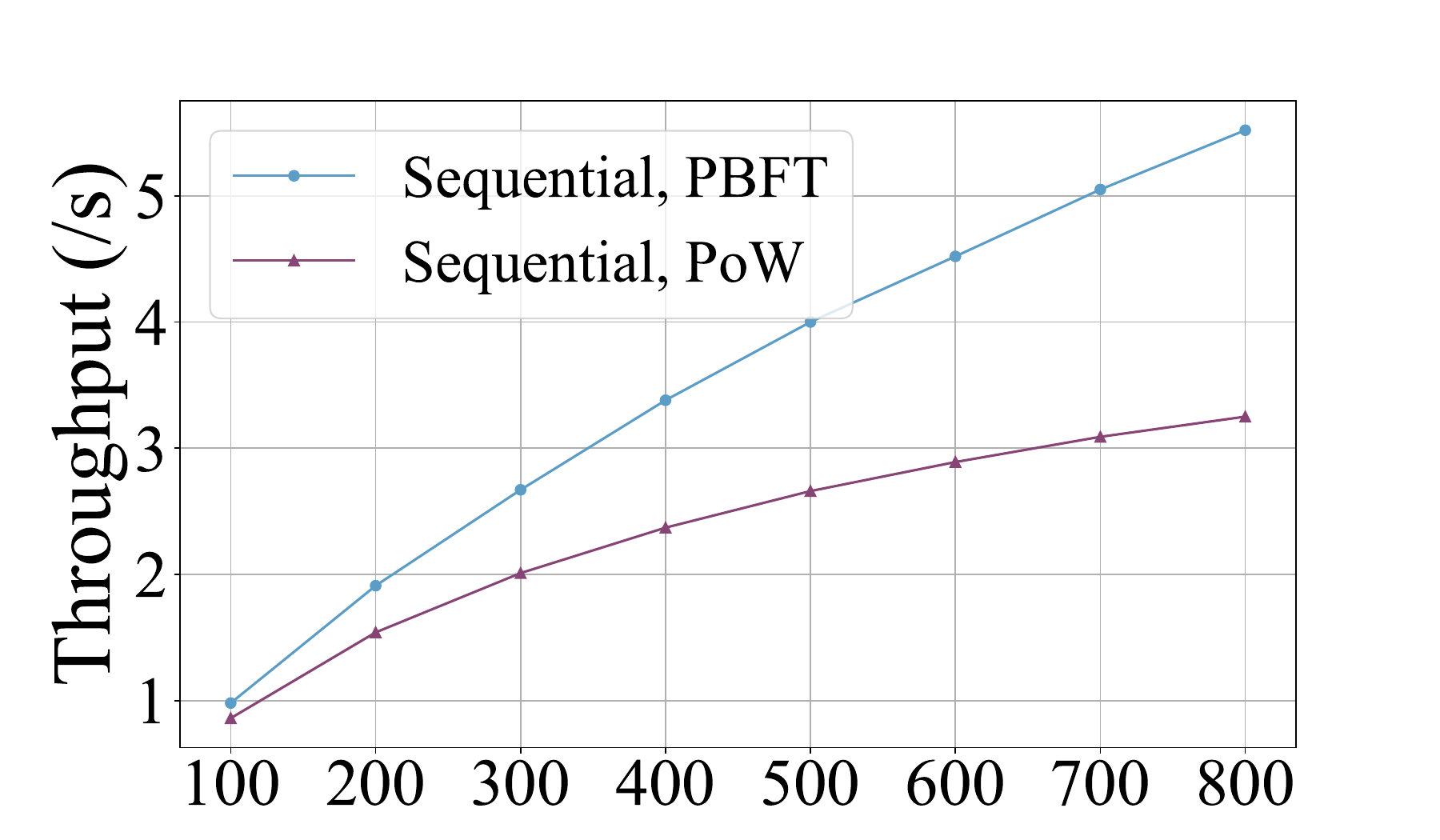}
        \label{fig:1Sequential}
    }
    \subfigure[]{
        \includegraphics[width=0.21\textwidth]{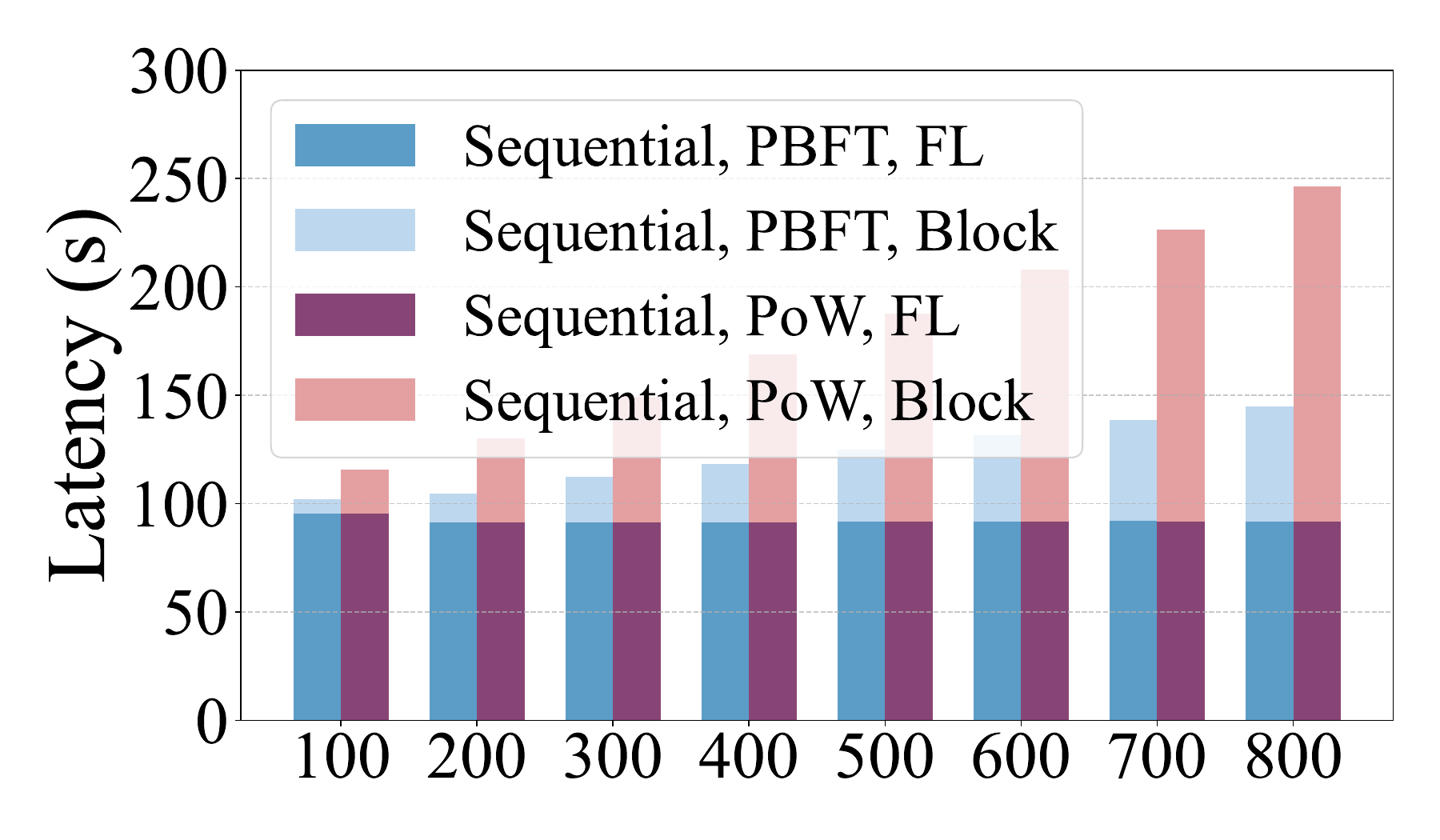}
        \label{fig:2Sequential}
    }
    \subfigure[]{
        \includegraphics[width=0.21\textwidth]{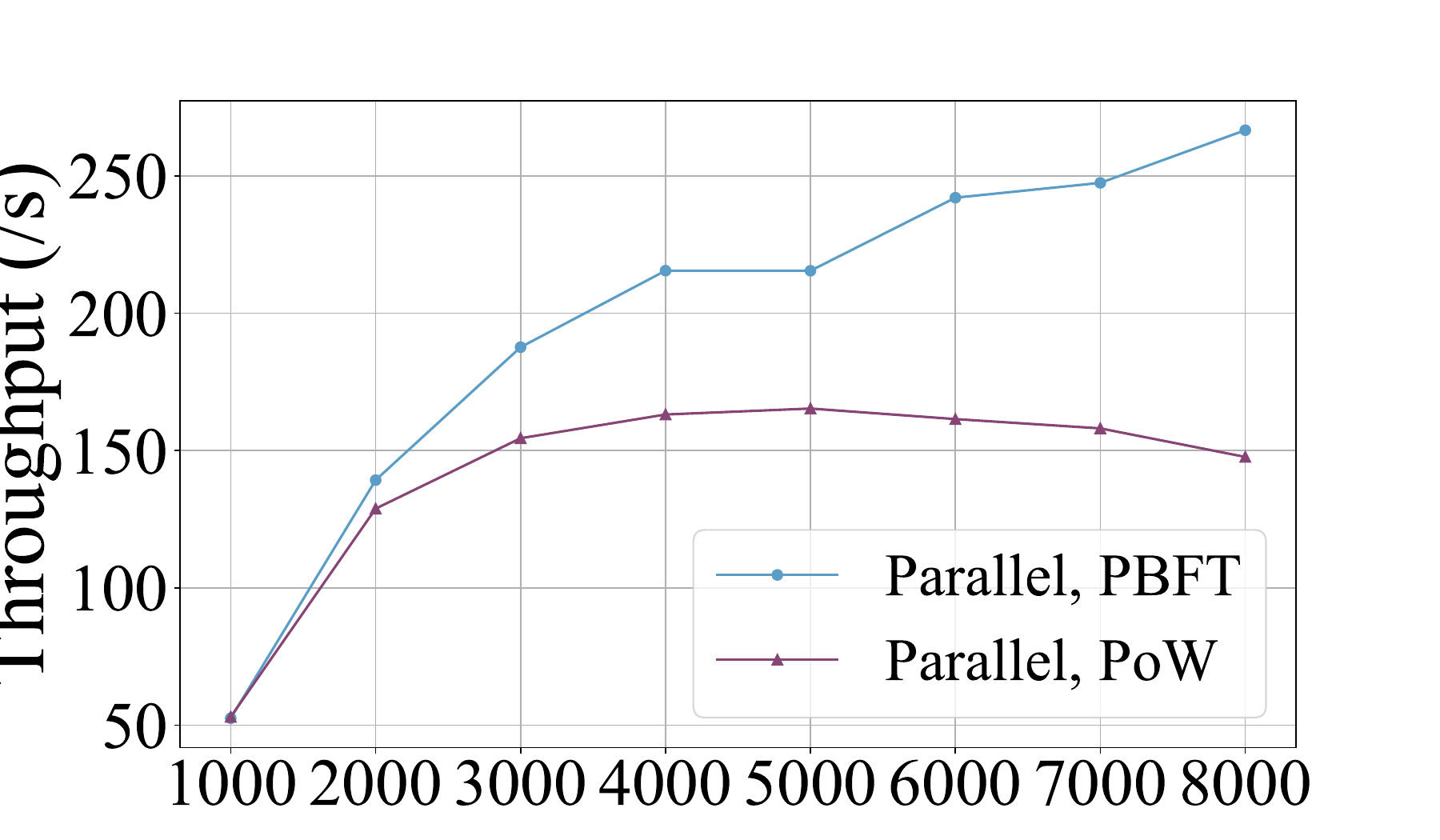}
        \label{fig:3Parallel}
    }
    \subfigure[]{
        \includegraphics[width=0.21\textwidth]{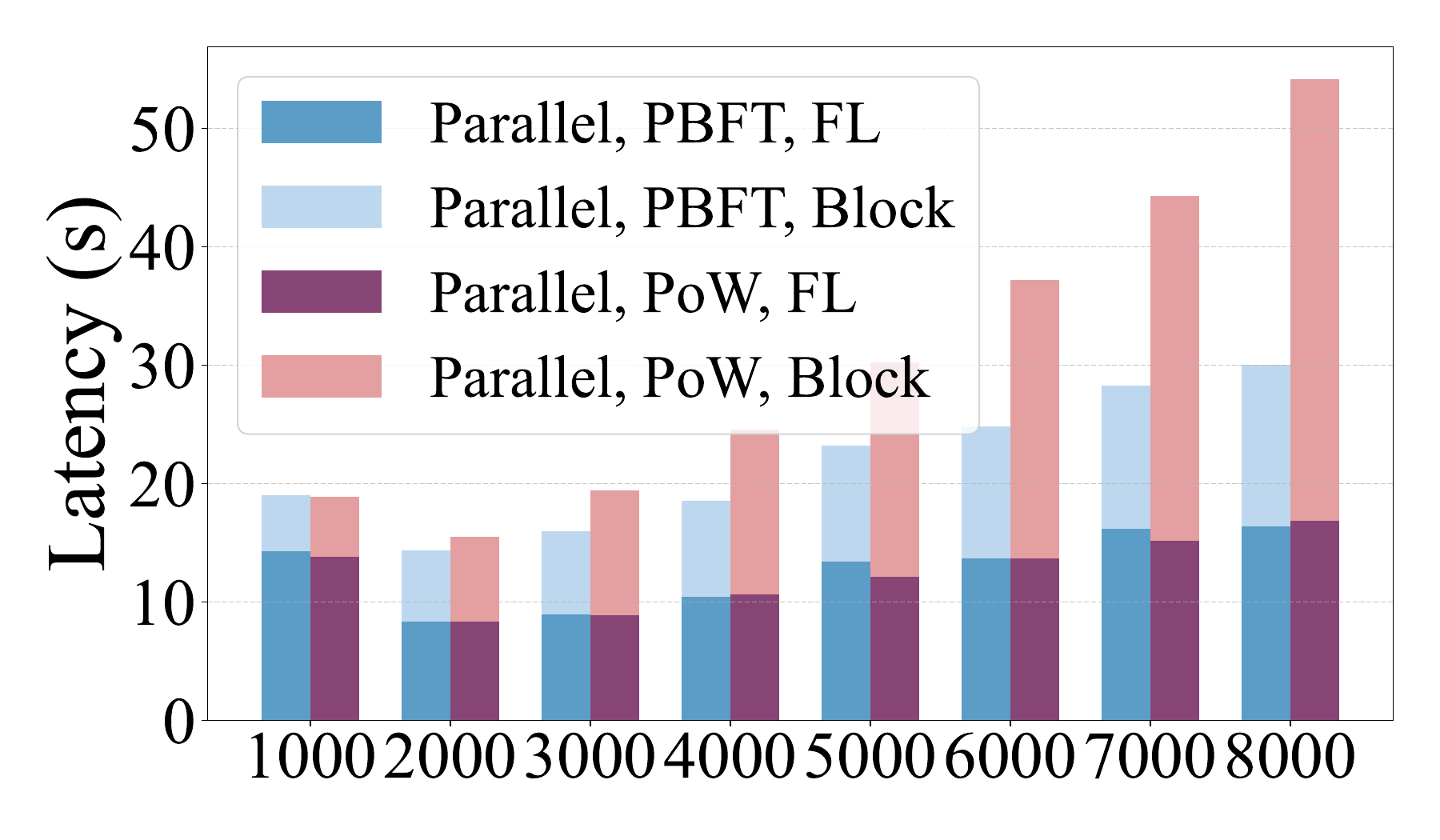}
        \label{fig:4Parallel}
    }
    \caption{\small (a) and (b) show the sequential throughput and latency; (c) and (d) show the parallel throughput and latency.}
    \label{fig:1-4}
\end{figure*}

\begin{figure*}[ht]
    \centering
    \subfigure[]{
        \includegraphics[width=0.21\textwidth]{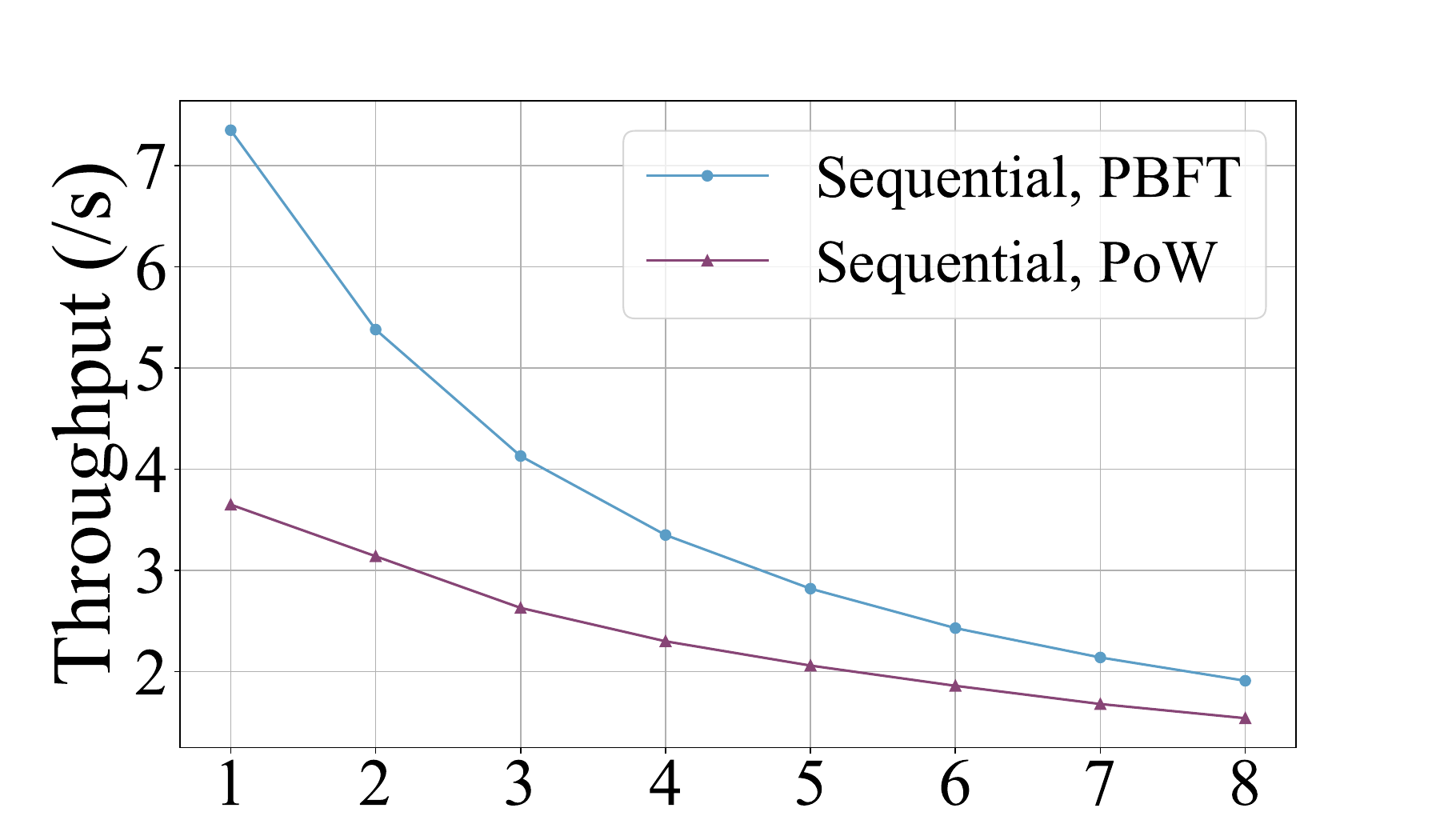}
        \label{fig:6Sequential}
    }
    \subfigure[]{
        \includegraphics[width=0.21\textwidth]{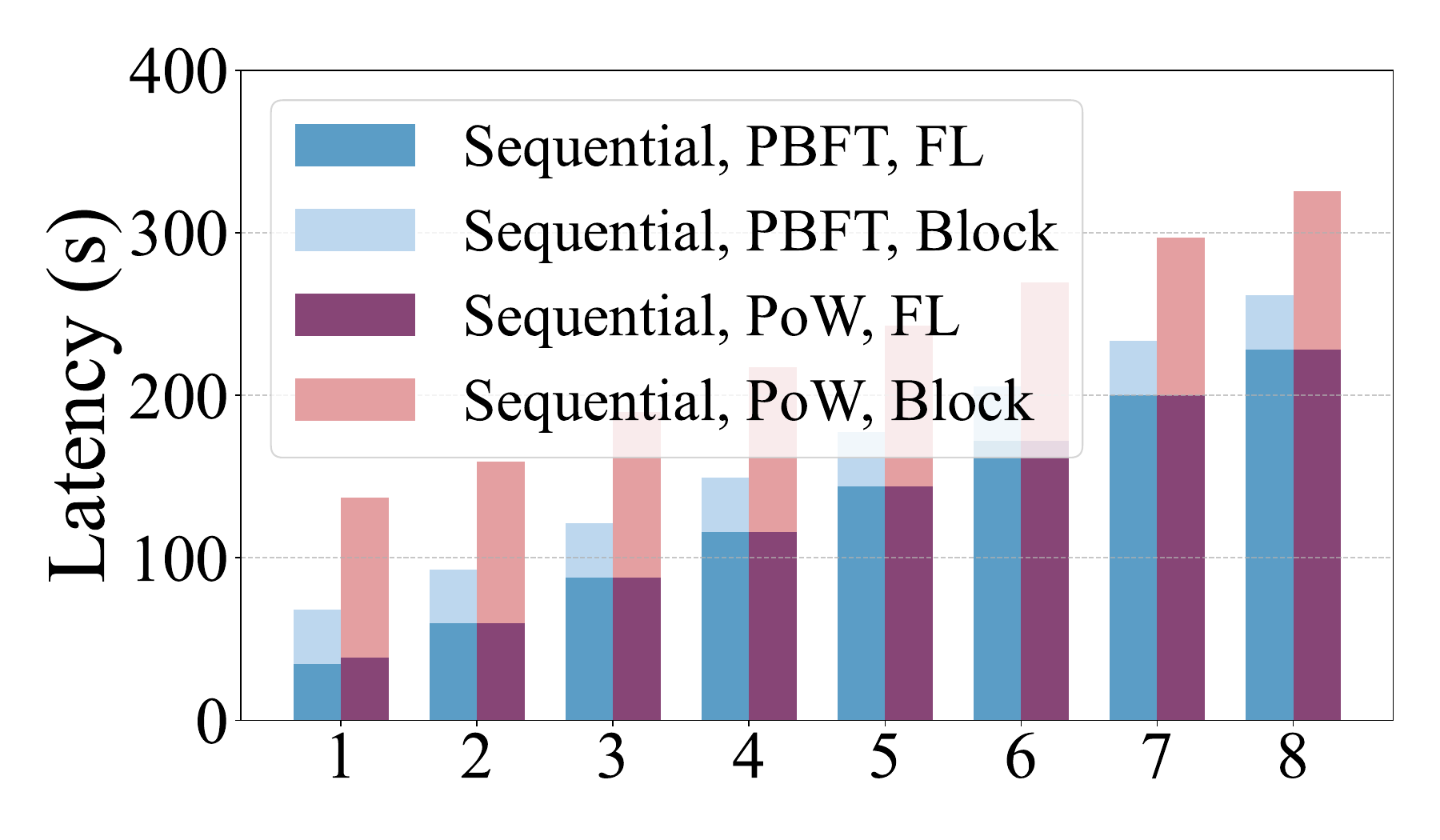}
        \label{fig:5Sequential}
    }
    \subfigure[]{
        \includegraphics[width=0.21\textwidth]{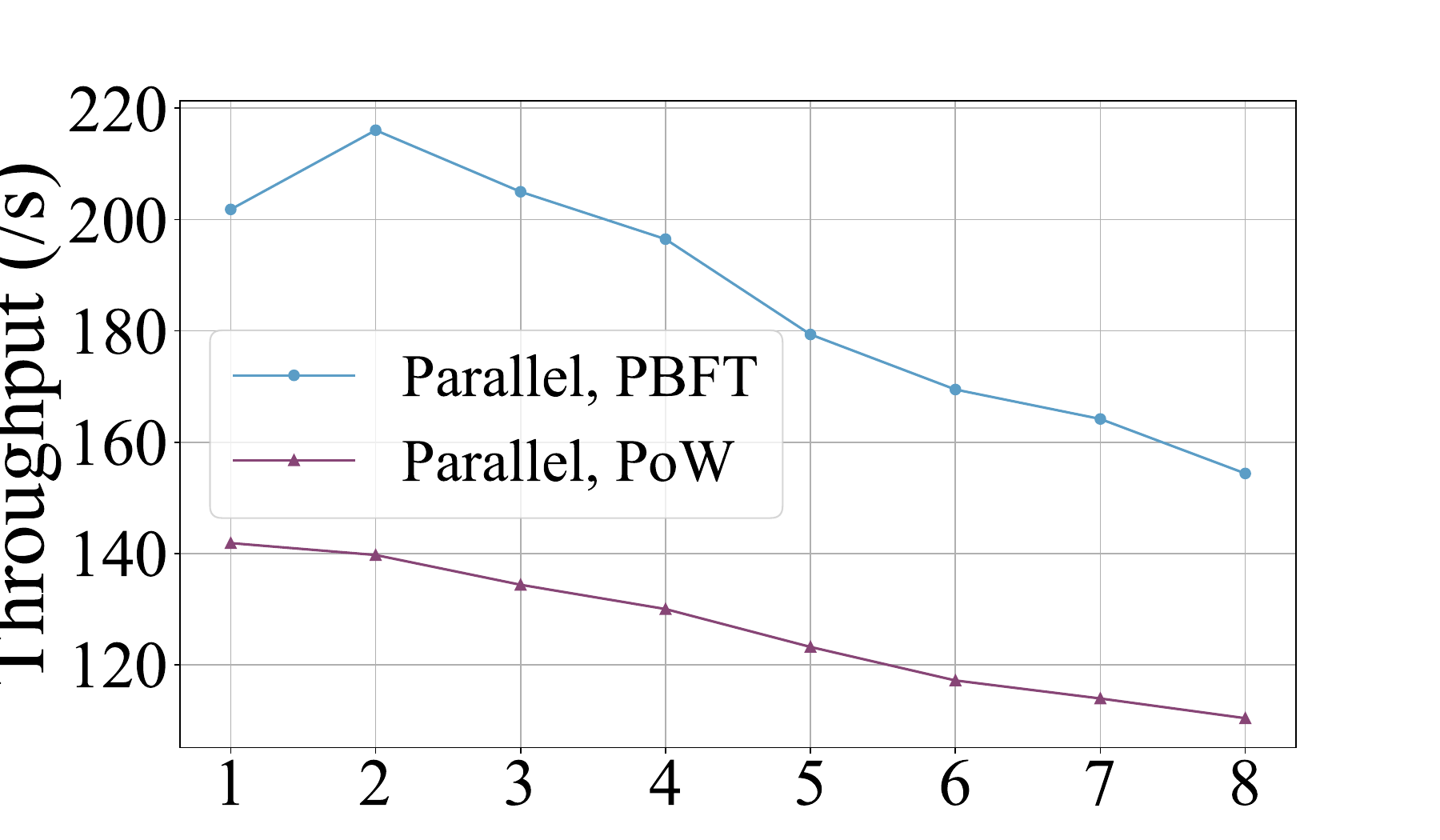}
        \label{fig:7Parallel}
    }
    \subfigure[]{
        \includegraphics[width=0.21\textwidth]{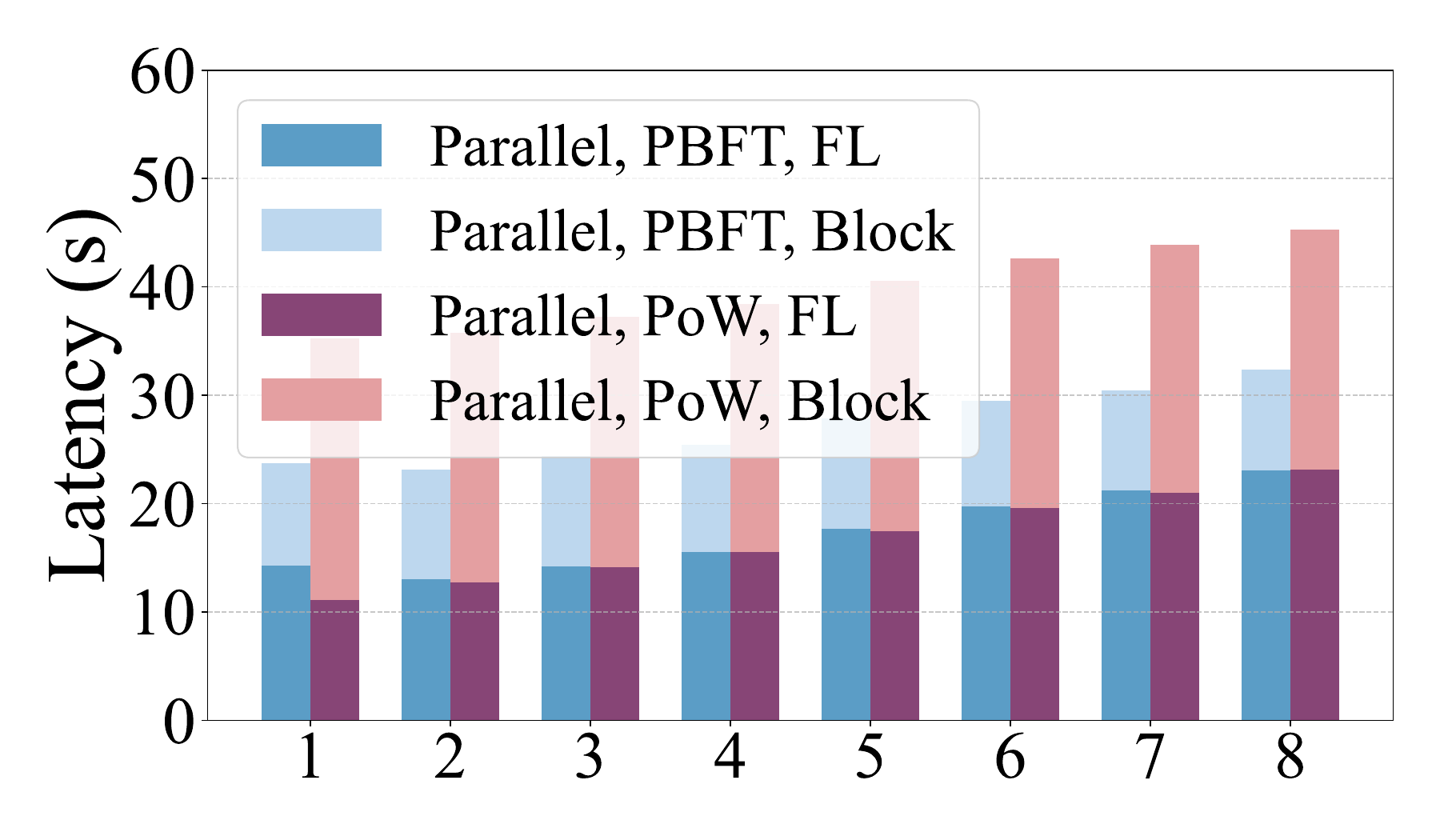}
        \label{fig:8Parallel}
    }
    \caption{\small (a) and (b) show the sequential throughput and latency across different DAG depths; (c) and (d) show the parallel throughput and latency across different DAG depths.}
    \label{fig:5-8}
\end{figure*}

\begin{figure*}[ht]
    \centering
    \subfigure[]{
        \includegraphics[width=0.21\textwidth]{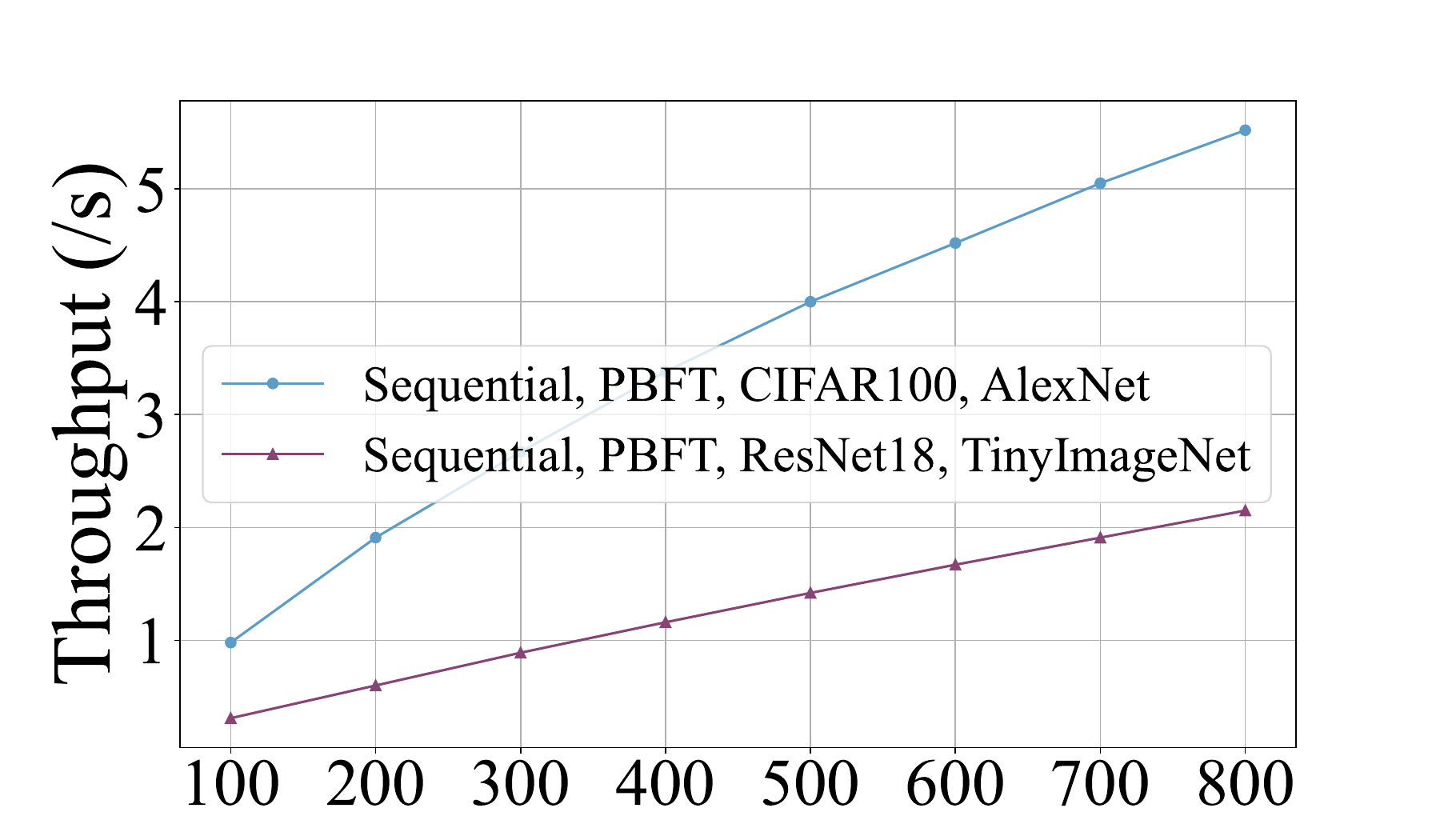}
        \label{fig:9Sequential}
    }
    \subfigure[]{
        \includegraphics[width=0.21\textwidth]{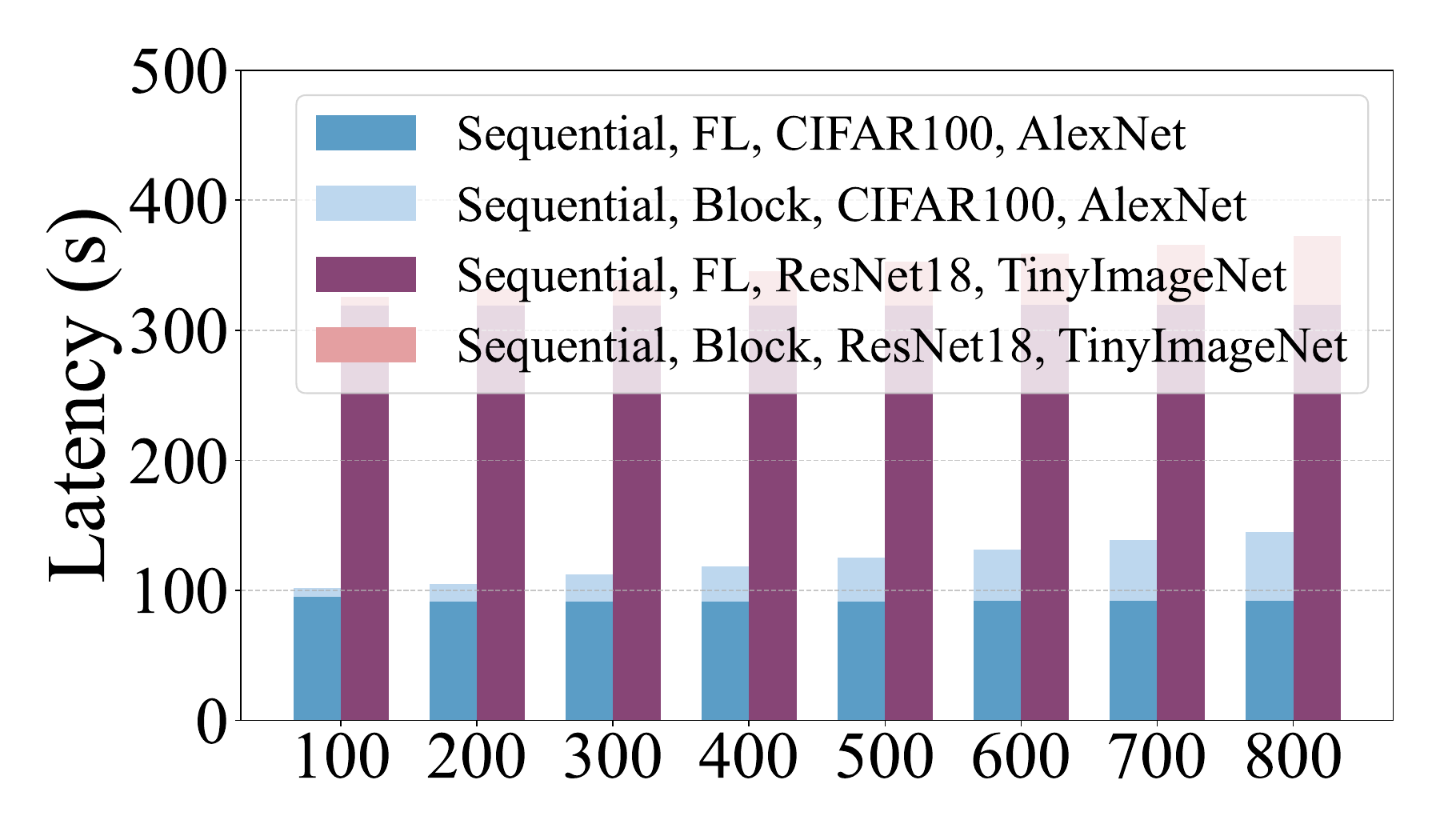}
        \label{fig:10Sequential}
    }
    \subfigure[]{
        \includegraphics[width=0.21\textwidth]{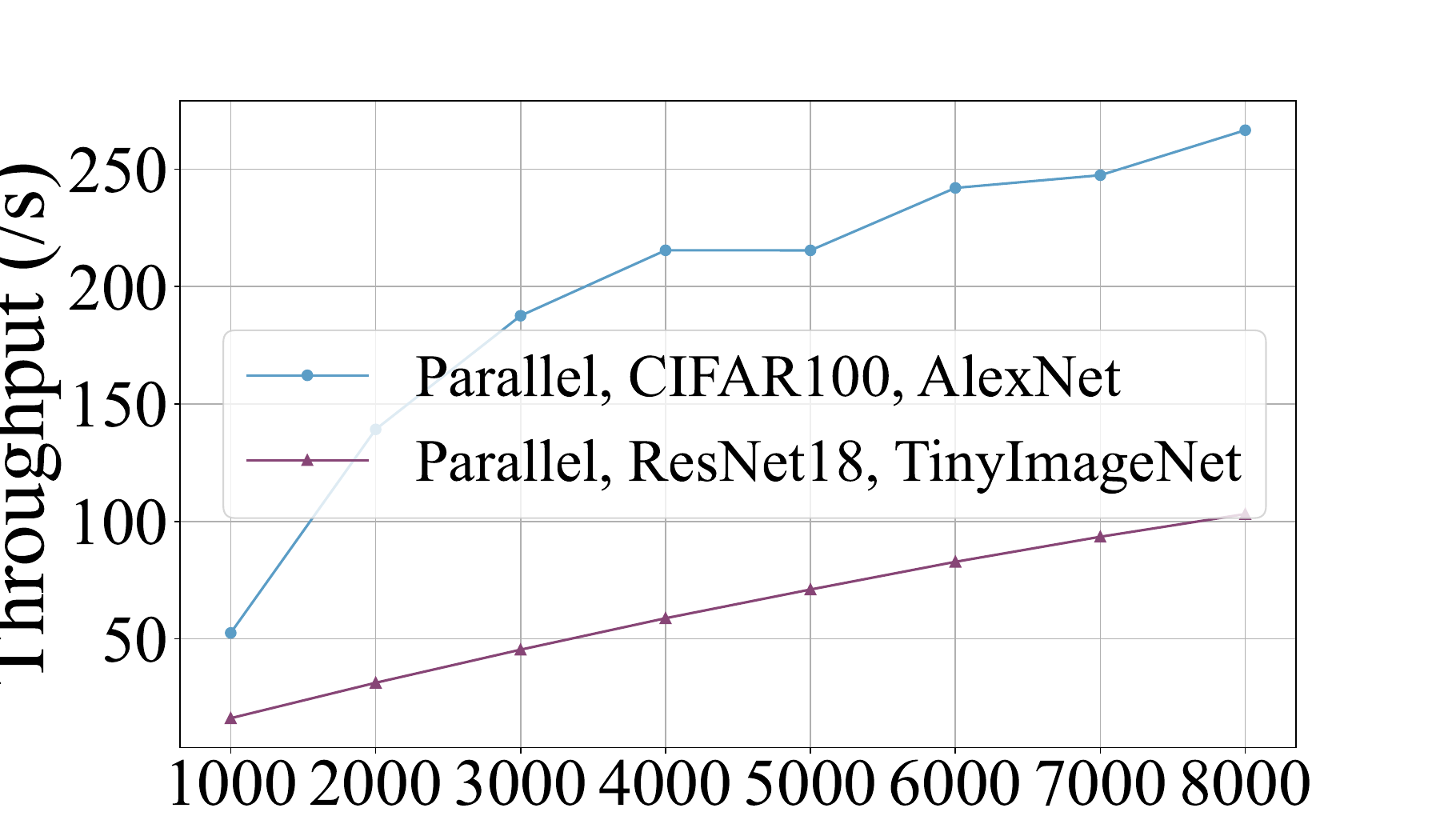}
        \label{fig:11Parallel}
    }
    \subfigure[]{
        \includegraphics[width=0.21\textwidth]{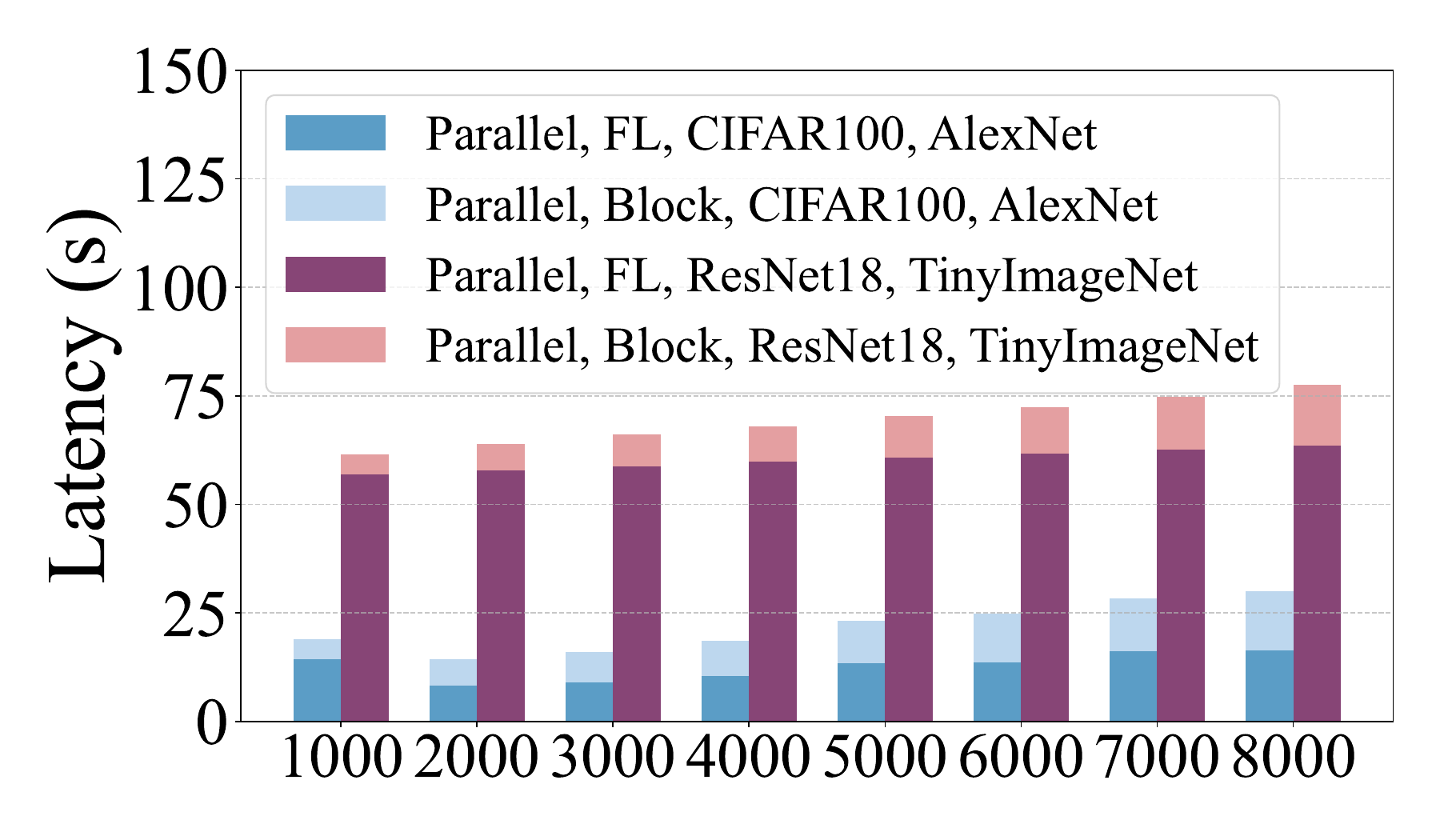}
        \label{fig:12Parallel}
    }
    \caption{\small (a) and (b) show the throughput and latency of sequential unlearning across different datasets and models; (c) and (d) show the throughput and latency of parallel unlearning across different datasets and models.}
    \label{fig:9-12}
\end{figure*}

We analyze FL and blockchain time overhead under sequential (re-training) and parallel (gradient ascent) paradigms for \textbf{G4}.
Each unlearning request maps to a model chain of depth 1 to 8, with each model linked to a block. We evaluate performance under sequential (100–800 requests) and parallel (1000–8000 requests) paradigms across varying scales and processing strategies.

In the sequential paradigm, throughput decreases as the number of unlearning requests grows, with both PBFT and PoW consensus mechanisms affected (Fig.~\ref{fig:1Sequential}). PoW experiences a steeper decline beyond 500 requests due to its higher block overhead, while PBFT sustains more stable throughput thanks to its lightweight consensus process. Latency, however, rises with increasing requests (Fig.~\ref{fig:2Sequential}), with PoW exhibiting the highest latency due to its complex operations, and PBFT with block requests achieving the lowest, reflecting its efficiency in handling sequential tasks.

The parallel paradigm significantly improves scalability, with throughput increasing as unlearning requests grow, due to concurrent handling of FL and block tasks (Fig.~\ref{fig:3Parallel}). PBFT outperforms PoW beyond 5000 requests by better managing block overhead, while PoW’s heavier consensus limits its performance. Latency drops with more requests (Fig.~\ref{fig:4Parallel}) as parallelism reduces processing time per request, though PoW still shows the highest latency compared to PBFT.

Moreover, DAG depth and model/dataset complexity have a non-neglible impact. In the sequential paradigm, throughput decreases with depth—PoW drops sharply beyond depth 5—while latency peaks at depth 8 (Figs.~\ref{fig:5Sequential}–\ref{fig:6Sequential}). In contrast, parallel processing yields higher throughput and lower latency, with PBFT consistently outperforming PoW, especially at depth 8 (Figs.~\ref{fig:7Parallel}–\ref{fig:8Parallel}). For different models and datasets, sequential throughput drops in complex cases like ResNet18 on TinyImageNet beyond 500 requests, with increased latency (Figs.~\ref{fig:9Sequential}–\ref{fig:10Sequential}), while parallel processing maintains lower latency and higher throughput, particularly for simpler setups like AlexNet on CIFAR100 after 5000 requests (Figs.~\ref{fig:11Parallel}–\ref{fig:12Parallel}).

\vspace{-0.4em}
\begin{center}
\fbox{%
\begin{minipage}{0.98\linewidth}
\change{\textbf{Takeaway—parallel paradigm scales efficiently with reduced latency.}
Parallel unlearning sig-
nificantly improves throughput and lowers latency under high request volume and deep model
inheritance, especially with high-speed consensus (e.g., PBFT). Unlike sequential re-training, which
suffers from increased overhead and latency as DAG depth and request count grow, gradient ascent
with parallel execution maintains stable performance and scalability across models, datasets, and
blockchain protocols.}
\end{minipage}
}
\end{center}


\section{Conclusion And Further Work}\label{section:6}
In response to \textbf{RQ1} on minimizing update costs for unlearning historical models in Blockchained FL, we developed \textit{BlockFUL}, a dual-chain framework ensuring traceable and trustworthy unlearning. For \textbf{RQ2} on adaptable unlearning methods, we introduced two paradigms in BlockFUL: parallel unlearning with gradient ascent and sequential unlearning with re-training. Our analysis showed that parallel unlearning is more cost-effective, particularly for consensus and transmission costs as model updates grow. Experiments revealed that sequential unlearning maintains high accuracy for retained data, while parallel unlearning effectiveness depends on inheritance depth and model variations. In future work, we plan to explore more efficient algorithms, such as model compression, to reduce computational costs for unlearning in large-scale datasets and complex models without compromising accuracy.

\bibliographystyle{IEEEtran}
\bibliography{main}

\begin{thebibliography}{10}
\providecommand{\url}[1]{#1}
\csname url@samestyle\endcsname
\providecommand{\newblock}{\relax}
\providecommand{\bibinfo}[2]{#2}
\providecommand{\BIBentrySTDinterwordspacing}{\spaceskip=0pt\relax}
\providecommand{\BIBentryALTinterwordstretchfactor}{4}
\providecommand{\BIBentryALTinterwordspacing}{\spaceskip=\fontdimen2\font plus
\BIBentryALTinterwordstretchfactor\fontdimen3\font minus \fontdimen4\font\relax}
\providecommand{\BIBforeignlanguage}[2]{{%
\expandafter\ifx\csname l@#1\endcsname\relax
\typeout{** WARNING: IEEEtran.bst: No hyphenation pattern has been}%
\typeout{** loaded for the language `#1'. Using the pattern for}%
\typeout{** the default language instead.}%
\else
\language=\csname l@#1\endcsname
\fi
#2}}
\providecommand{\BIBdecl}{\relax}
\BIBdecl

\bibitem{uddin2021blockchain}
M.~Uddin, K.~Salah, R.~Jayaraman, S.~Pesic, and S.~Ellahham, ``Blockchain for drug traceability: Architectures and open challenges,'' \emph{Health Informatics Journal}, vol.~27, no.~2, p. 14604582211011228, 2021.

\bibitem{li2020blockchain}
Y.~Li, C.~Chen, N.~Liu, H.~Huang, Z.~Zheng, and Q.~Yan, ``A blockchain-based decentralized federated learning framework with committee consensus,'' \emph{IEEE Network}, vol.~35, no.~1, pp. 234--241, 2020.

\bibitem{krauss2024automatic}
T.~Krau{\ss}, J.~K{\"o}nig, A.~Dmitrienko, and C.~Kanzow, ``Automatic adversarial adaption for stealthy poisoning attacks in federated learning,'' in \emph{To appear soon at the Network and Distributed System Security Symposium (NDSS)}, 2024.

\bibitem{chu2024multi}
Y.-W. Chu, S.~Hosseinalipour, E.~Tenorio, L.~Cruz, K.~Douglas, A.~S. Lan, and C.~G. Brinton, ``Multi-layer personalized federated learning for mitigating biases in student predictive analytics,'' \emph{IEEE Transactions on Emerging Topics in Computing}, 2024.

\bibitem{nguyen2022latency}
D.~C. Nguyen, S.~Hosseinalipour, D.~J. Love, P.~N. Pathirana, and C.~G. Brinton, ``Latency optimization for blockchain-empowered federated learning in multi-server edge computing,'' \emph{IEEE Journal on Selected Areas in Communications}, vol.~40, no.~12, pp. 3373--3390, 2022.

\bibitem{yu2023ironforge}
G.~Yu, X.~Wang, C.~Sun, Q.~Wang, P.~Yu, W.~Ni, and R.~P. Liu, ``Ironforge: An open, secure, fair, decentralized federated learning,'' \emph{IEEE Transactions on Neural Networks and Learning Systems}, 2023.

\bibitem{halimi2022federated}
A.~Halimi, S.~Kadhe, A.~Rawat, and N.~Baracaldo, ``Federated unlearning: How to efficiently erase a client in fl?'' \emph{arXiv preprint arXiv:2207.05521}, 2022.

\bibitem{yuan2023federated}
W.~Yuan, H.~Yin, F.~Wu, S.~Zhang, T.~He, and H.~Wang, ``Federated unlearning for on-device recommendation,'' in \emph{Proceedings of the Sixteenth ACM International Conference on Web Search and Data Mining}, 2023, pp. 393--401.

\bibitem{chundawat2023zero}
V.~S. Chundawat, A.~K. Tarun, M.~Mandal, and M.~Kankanhalli, ``Zero-shot machine unlearning,'' \emph{IEEE Transactions on Information Forensics and Security}, vol.~18, pp. 2345--2354, 2023.

\bibitem{yu2020enabling}
G.~Yu, X.~Zha, X.~Wang, W.~Ni, K.~Yu, P.~Yu, J.~A. Zhang, R.~P. Liu, and Y.~J. Guo, ``Enabling attribute revocation for fine-grained access control in blockchain-iot systems,'' \emph{IEEE Transactions on Engineering Management}, vol.~67, no.~4, pp. 1213--1230, 2020.

\bibitem{bourtoule2021machine}
L.~Bourtoule, V.~Chandrasekaran, C.~A. Choquette-Choo, H.~Jia, A.~Travers, B.~Zhang, D.~Lie, and N.~Papernot, ``Machine unlearning,'' in \emph{2021 IEEE Symposium on Security and Privacy (SP)}.\hskip 1em plus 0.5em minus 0.4em\relax IEEE, 2021, pp. 141--159.

\bibitem{koch2023no}
K.~Koch and M.~Soll, ``No matter how you slice it: Machine unlearning with sisa comes at the expense of minority classes,'' in \emph{2023 IEEE Conference on Secure and Trustworthy Machine Learning (SaTML)}.\hskip 1em plus 0.5em minus 0.4em\relax IEEE, 2023, pp. 622--637.

\bibitem{liu2020federated}
G.~Liu, X.~Ma, Y.~Yang, C.~Wang, and J.~Liu, ``Federated unlearning,'' \emph{arXiv preprint arXiv:2012.13891}, 2020.

\bibitem{liu2022right}
Y.~Liu, L.~Xu, X.~Yuan, C.~Wang, and B.~Li, ``The right to be forgotten in federated learning: An efficient realization with rapid retraining,'' in \emph{IEEE INFOCOM 2022-IEEE Conference on Computer Communications}.\hskip 1em plus 0.5em minus 0.4em\relax IEEE, 2022, pp. 1749--1758.

\bibitem{liu2022backdoor}
Y.~Liu, M.~Fan, C.~Chen, X.~Liu, Z.~Ma, L.~Wang, and J.~Ma, ``Backdoor defense with machine unlearning,'' in \emph{IEEE INFOCOM 2022-IEEE Conference on Computer Communications}.\hskip 1em plus 0.5em minus 0.4em\relax IEEE, 2022, pp. 280--289.

\bibitem{jia2022redactable}
M.~Jia, J.~Chen, K.~He, R.~Du, L.~Zheng, M.~Lai, D.~Wang, and F.~Liu, ``Redactable blockchain from decentralized chameleon hash functions,'' \emph{IEEE Transactions on Information Forensics and Security}, vol.~17, pp. 2771--2783, 2022.

\bibitem{wu2022federated}
C.~Wu, S.~Zhu, and P.~Mitra, ``Federated unlearning with knowledge distillation,'' \emph{arXiv preprint arXiv:2201.09441}, 2022.

\bibitem{wu2022federated1}
L.~Wu, S.~Guo, J.~Wang, Z.~Hong, J.~Zhang, and Y.~Ding, ``Federated unlearning: Guarantee the right of clients to forget,'' \emph{IEEE Network}, vol.~36, no.~5, pp. 129--135, 2022.

\bibitem{che2023fast}
T.~Che, Y.~Zhou, Z.~Zhang, L.~Lyu, J.~Liu, D.~Yan, D.~Dou, and J.~Huan, ``Fast federated machine unlearning with nonlinear functional theory,'' in \emph{International Conference on Machine Learning}.\hskip 1em plus 0.5em minus 0.4em\relax PMLR, 2023, pp. 4241--4268.

\bibitem{wang2022federated}
J.~Wang, S.~Guo, X.~Xie, and H.~Qi, ``Federated unlearning via class-discriminative pruning,'' in \emph{Proceedings of the ACM Web Conference 2022}, 2022, pp. 622--632.

\bibitem{su2023asynchronous}
N.~Su and B.~Li, ``Asynchronous federated unlearning,'' in \emph{IEEE INFOCOM 2023-IEEE Conference on Computer Communications}.\hskip 1em plus 0.5em minus 0.4em\relax IEEE, 2023, pp. 1--10.

\bibitem{islam2024federated}
A.~Islam, H.~Karimipour, T.~R. Gadekallu, and Y.~Zhu, ``A federated unlearning-based secure management scheme to enable automation in smart consumer electronics facilitated by digital twin,'' \emph{IEEE Transactions on Consumer Electronics}, 2024.

\bibitem{lin2024blockchain}
Y.~Lin, Z.~Gao, H.~Du, J.~Ren, Z.~Xie, and D.~Niyato, ``Blockchain-enabled trustworthy federated unlearning,'' \emph{arXiv preprint arXiv:2401.15917}, 2024.

\bibitem{zuo2024federated}
X.~Zuo, M.~Wang, T.~Zhu, L.~Zhang, S.~Yu, and W.~Zhou, ``Federated learning with blockchain-enhanced machine unlearning: A trustworthy approach,'' \emph{arXiv preprint arXiv:2405.20776}, 2024.

\bibitem{voigt2017eu}
P.~Voigt and A.~Von~dem Bussche, ``The {EU} general data protection regulation ({GDPR}),'' \emph{A Practical Guide, 1st Ed., Cham: Springer International Publishing}, vol.~10, no. 3152676, pp. 10--5555, 2017.

\bibitem{abadi2016deep}
M.~Abadi, A.~Chu, I.~Goodfellow, H.~B. McMahan, I.~Mironov, K.~Talwar, and L.~Zhang, ``Deep learning with differential privacy,'' in \emph{Proceedings of the 2016 ACM SIGSAC conference on computer and communications security}, 2016, pp. 308--318.

\bibitem{gavzi2019proof}
P.~Ga{\v{z}}i, A.~Kiayias, and D.~Zindros, ``Proof-of-stake sidechains,'' in \emph{2019 IEEE Symposium on Security and Privacy (SP)}.\hskip 1em plus 0.5em minus 0.4em\relax IEEE, 2019, pp. 139--156.

\bibitem{ganeriwal2008reputation}
S.~Ganeriwal, L.~K. Balzano, and M.~B. Srivastava, ``Reputation-based framework for high integrity sensor networks,'' \emph{ACM Transactions on Sensor Networks (TOSN)}, vol.~4, no.~3, pp. 1--37, 2008.

\bibitem{gilad2017algorand}
Y.~Gilad, R.~Hemo, S.~Micali, G.~Vlachos, and N.~Zeldovich, ``Algorand: Scaling {B}yzantine agreements for cryptocurrencies,'' in \emph{Proceedings of the 26th Symposium on Operating Systems Principles}, 2017, pp. 51--68.

\bibitem{lastlayer}
H.~Kim, S.~Lee, and S.~S. Woo, ``Layer attack unlearning: Fast and accurate machine unlearning via layer level attack and knowledge distillation,'' in \emph{Proceedings of the AAAI Conference on Artificial Intelligence}, vol.~38, no.~19, 2024, pp. 21\,241--21\,248.

\bibitem{li2017learning}
Z.~Li and D.~Hoiem, ``Learning without forgetting,'' \emph{IEEE Transactions on Pattern Analysis and Machine Intelligence}, vol.~40, no.~12, pp. 2935--2947, 2017.

\end{thebibliography}

\vfill

\end{document}